\documentclass[11pt]{article}
\usepackage[margin=1in]{geometry}
\usepackage{amsmath}
\usepackage{appendix}
\usepackage{amsthm}
\usepackage{bbm}
\usepackage{amssymb}
\usepackage{setspace}
\usepackage{color}
\usepackage{float}
\usepackage{dsfont}
\usepackage{enumerate}
\usepackage{hyperref}
\usepackage[capitalize]{cleveref}
\usepackage{graphicx}
\usepackage{algpseudocode,algorithm,algorithmicx}
\usepackage{tikz}
\usepackage{subcaption} 
\usepackage{wrapfig}
\usepackage{makecell}
\usepackage{mathdots} 

\newtheorem{theorem}{Theorem}
\newtheorem{proposition}[theorem]{Proposition}
\newtheorem{lemma}[theorem]{Lemma}

\newtheorem{observation}[theorem]{Observation}

\newtheorem{assumption}[theorem]{Assumption}
\theoremstyle{definition}

\newtheorem{remark}[theorem]{Remark}
\numberwithin{theorem}{section} 

\numberwithin{equation}{section}

\usepackage{setspace}

\def\({\Big(}
\def\){\Big)}
\def\C{\mathbb{C}}
\def\D{\mathbb{D}}
\def\E{\mathbb{E}}

\def\R{\mathbb{R}}

\def\T{\mathbb{T}}

\def\Z{\mathbb{Z}}

\renewcommand{\subset}{\subseteq}

\renewcommand{\tilde}{\widetilde}
\renewcommand{\epsilon}{\varepsilon}

\def\range{{\rm range}}

\def\tr{{\rm tr}}

\newcommand{\commentout}[1]{}
\newcommand{\bfa}{{\boldsymbol a}}

\newcommand{\bfu}{{\boldsymbol u}}
\newcommand{\bfv}{{\boldsymbol v}}

\newcommand{\bfy}{{\boldsymbol y}}

\newcommand{\bfA}{{\boldsymbol A}}

\newcommand{\bfI}{{\boldsymbol I}}

\newcommand{\bfN}{{\boldsymbol N}}

\newcommand{\bfU}{{\boldsymbol U}}

\newcommand{\bfW}{{\boldsymbol W}}
\newcommand{\bfY}{{\boldsymbol Y}}

\newcommand{\bfzero}{{\boldsymbol 0}}

\newcommand{\bfeta}{{\boldsymbol \eta}}

\newcommand{\bfphi}{{\boldsymbol \phi}}
\newcommand{\bfPhi}{{\boldsymbol \Phi}}
\newcommand{\bfpsi}{{\boldsymbol \psi}}

\newcommand{\bfSigma}{{\boldsymbol \Sigma}}

\newcommand{\calN}{\mathcal{N}}

\def\forspace{\quad\text{for}\quad}

\def\forallspace{\quad\text{for all}\quad}

\def\andspace{\quad\text{and}\quad}

\numberwithin{equation}{section}

\title{A sharp analysis of Root-MUSIC: \\ locations of correct and extraneous roots}
\author{Hana Huber\footnote{CUNY Graduate Center, New York, NY, USA. Email: hhuber@gc.cuny.edu} \and Weilin Li\footnote{Corresponding author. CUNY Graduate Center and City College, New York, NY, USA. Email: wli6@ccny.cuny.edu}}

\begin{document}
	
\maketitle 

\begin{abstract}

    Root-MUSIC is a spectral estimation algorithm that approximates the unknown signal frequencies by constructing a high-degree polynomial and finding a subset of roots which are closest to the complex unit circle. Previous works found asymptotic expectation formulas for the performance of Root-MUSIC under the implicit assumption that the aforementioned root selection criterion does not select extraneous roots -- those which are unrelated to the correct parameters. This paper removes the need for this assumption by showing all extraneous roots lie outside an annulus of a certain thickness and therefore are not selected by the algorithm. This paper also provides sharp and non-asymptotic error bounds for the correct roots in terms of fundamental model parameters. All results hold under a natural separation condition on the correct signal frequencies and are applicable in both the single- and multi-snapshot models. Numerical simulations confirm our theory. The main mathematical insight of this paper is a geometric property of the Root-MUSIC polynomial: its correct roots are highly stable to noise while its extraneous roots must lie outside of an annulus. 
    
\end{abstract}

\medskip
\noindent
{\bf 2020 Math Subject Classification:} 94A12, 30C15, 42A10, 42A82, 68W40


\section{Introduction} 

\subsection{Motivation} 

A typical spectral estimation problem is to recover the frequencies $\{\theta_k\}_{k=1}^s\subset [0,2\pi)$ of a Fourier sum from noisy temporal samples. There are two standard models of this nonlinear inverse problem, which are explained in \cref{sec:subspace}. A collection of highly celebrated ``subspace methods" \cite{schmidt1986multiple,barabell1983improving,kailath1989esprit,hua1990matrixpencil} were created for this task. They start by approximating the correct ``signal subspace" $\bfU$ that is associated with the correct frequencies $\{\theta_k\}_{k=1}^s$. The correct signal subspace $\bfU$ has dimension $s$ in $\C^m$, where $m$ is usually proportional to the number of collected samples. The analytical and statistical theories of subspace estimation are fairly well-understood, see \cite{stewart1990matrix,chen2021spectral} for overviews. Once an approximate or empirical signal subspace $\tilde \bfU$ is computed using the observations, each subspace method carries out a different routine to process $\tilde\bfU$ to produce estimated signal frequencies, which hopefully approximate the correct signal frequencies. 

The Root-MUSIC algorithm \cite{barabell1983improving} is a variation of the classical MUSIC algorithm \cite{schmidt1981signal,stoica1989music}. Both are foundational subspace methods which are readily available to a general audience -- for example, both are included in MATLAB's Signal Processing Toolbox. While the classical MUSIC algorithm finds the $s$ smallest local minima of a particular ``MUSIC function" through a brute force search, Root-MUSIC was developed as a computationally faster alternative, which we will explain momentarily. 

The Root-MUSIC algorithm uses an approximate subspace $\tilde\bfU$ to define a ``Root-MUSIC polynomial" $\tilde P$ of degree at most $2m-2$. Its roots away from zero come in pairs of the form $(w,1/\overline w)$, with $w\not=0$ inside the closed unit disk $\overline\D\subset\C$ and its reflection $1/\overline w\not\in \D$. Root-MUSIC selects the $s$ roots of $\tilde P$ (mod reflection) closest to the unit circle, denoted $\{\tilde w_k\}_{k=1}^s$, and $\{\tilde\theta_k:=\arg(\tilde w_k)\}_{k=1}^s$ serves as the estimated signal frequencies. Since $w$ and its reflection $1/\overline w$ have the same argument, we ignore all reflected roots to simplify subsequent discussion. Root-MUSIC is significantly faster than classical MUSIC due to the former's ability to use fast root-finding procedures such as Newton's method instead of doing an expensive brute force search.

\begin{figure}[ht]
    \centering 
    \begin{subfigure}{0.44\textwidth}
    \begin{tikzpicture} [scale=0.7]
    \draw[very thick] (0,0) circle (3); 
    \draw[gray, thick, <->] (0,-5) -- (0,5) coordinate (im axis); 
    \draw[gray, thick, <->] (-5,0) -- (5,0) coordinate (re axis); 
    
    \fill[red] (0,0) ++(60:3) circle[radius=1pt] coordinate (a);
    \fill[red] (0,0) ++(100:3) circle[radius=1pt] coordinate (b);
    \fill[red] (0,0) ++(160:3) circle[radius=1pt] coordinate (c); 
    \fill[red] (0,0) ++(300:3) circle[radius=1pt] coordinate (e);
    \draw[red, very thick] (a)+(-0.2,-0.2) -- (a);
    \draw[red, very thick] (a)+(0.2,0.2) node [above] {$e^{i\theta_1}$} -- (a);
    \draw[red, very thick] (a)+(-0.2,0.2) -- (a);
    \draw[red, very thick] (a)+(0.2,-0.2) -- (a);
    \draw[red, very thick] (b)+(-0.2,-0.2) -- (b);
    \draw[red, very thick] (b)+(0.2,0.2) -- (b);
    \draw[red, very thick] (b)+(-0.2,0.2) node [above] {$e^{i\theta_2}$} -- (b);
    \draw[red, very thick] (b)+(0.2,-0.2) -- (b);
    \draw[red, very thick] (c)+(-0.2,-0.2) -- (c);
    \draw[red, very thick] (c)+(0.2,0.2) -- (c);
    \draw[red, very thick] (c)+(-0.2,0.2) node [left] {$e^{i\theta_3}$} -- (c);
    \draw[red, very thick] (c)+(0.2,-0.2) -- (c);
    \draw[red, very thick] (e)+(-0.2,-0.2) -- (e);
    \draw[red, very thick] (e)+(0.2,0.2) -- (e);
    \draw[red, very thick] (e)+(-0.2,0.2) -- (e);
    \draw[red, very thick] (e)+(0.2,-0.2) node [below] {$e^{i\theta_4}$} -- (e);
    
    \fill[blue] (0,0) ++(20:3.7) circle[radius=3pt] node [above=2pt] {$1/\overline {w_5}$};
    \fill[blue] (0,0) ++(20:2.6) circle[radius=3pt] node [below left] {$w_5$};
    \fill[blue] (0,0) ++(135:4) circle[radius=3pt] node [above left] {$1/\overline {w_6}$};
    \fill[blue] (0,0) ++(135:2.3) circle[radius=3pt] node [below right] {$w_6$};
    \fill[blue] (0,0) ++(200:4.2) circle[radius=3pt] node [below=2pt] {$1/\overline {w_7}$};
    \fill[blue] (0,0) ++(200:2.1) circle[radius=3pt] node [right=2pt] {$w_7$};
    \fill[blue] (0,0) ++(230:3.45) circle[radius=3pt] node [below=2pt] {$1/\overline {w_8}$};
    \fill[blue] (0,0) ++(230:2.85) circle[radius=3pt] node [above right] {$w_8$};
    \fill[blue] (0,0) ++(340:4.3) circle[radius=3pt] node [below=2pt] {$1/\overline {w_9}$};
    \fill[blue] (0,0) ++(340:2) circle[radius=3pt] node [left=2pt] {$w_9$};
\end{tikzpicture} 
\caption{Roots of $P$.}    
    \end{subfigure}
\begin{subfigure}{0.44\textwidth}
    \begin{tikzpicture} [scale=0.7]
    \draw[very thick] (0,0) circle (3); 
    \draw[gray, thick, <->] (0,-5) -- (0,5) coordinate (im axis); 
    \draw[gray, thick, <->] (-5,0) -- (5,0) coordinate (re axis); 
    
    \fill[blue] (0,0) ++(55:3.7) circle[radius=3pt];
    \fill[blue] (0,0) ++(55:2.6) circle[radius=3pt] node [below left] {$\tilde w_1$};
    \fill[red] (0,0) ++(103:3.2) circle[radius=1pt] coordinate (a);
    \draw[red, very thick] (a)+(-0.2,-0.2) -- (a);
    \draw[red, very thick] (a)+(0.2,0.2) -- (a);
    \draw[red, very thick] (a)+(-0.2,0.2) -- (a);
    \draw[red, very thick] (a)+(0.2,-0.2) -- (a);
    \fill[red] (0,0) ++(103:2.95) circle[radius=1pt] coordinate (b);
    \draw[red, very thick] (b)+(-0.2,-0.2) -- (b);
    \draw[red, very thick] (b)+(0.2,0.2) -- (b);
    \draw[red, very thick] (b)+(-0.2,0.2) -- (b);
    \draw[red, very thick] (b)+(0.2,-0.2) node [below] {$\tilde w_2$} -- (b);
    \fill[red] (0,0) ++(158:3.3) circle[radius=1pt] coordinate (c);
    \draw[red, very thick] (c)+(-0.2,-0.2) -- (c);
    \draw[red, very thick] (c)+(0.2,0.2) -- (c);
    \draw[red, very thick] (c)+(-0.2,0.2) -- (c);
    \draw[red, very thick] (c)+(0.2,-0.2) -- (c);
    \fill[red] (0,0) ++(158:2.9) circle[radius=1pt] coordinate (d);
    \draw[red, very thick] (d)+(-0.2,-0.2) -- (d);
    \draw[red, very thick] (d)+(0.2,0.2) -- (d);
    \draw[red, very thick] (d)+(-0.2,0.2) -- (d);
    \draw[red, very thick] (d)+(0.2,-0.2) -- (d) node [below right] {$\tilde w_3$};
    \fill[red] (0,0) ++(308:3.4) circle[radius=1pt] coordinate (e);
    \draw[red, very thick] (e)+(-0.2,-0.2) -- (e);
    \draw[red, very thick] (e)+(0.2,0.2) -- (e);
    \draw[red, very thick] (e)+(-0.2,0.2) -- (e);
    \draw[red, very thick] (e)+(0.2,-0.2) -- (e);
    \fill[red] (0,0) ++(308:2.8) circle[radius=1pt] coordinate (f);
    \draw[red, very thick] (f)+(-0.2,-0.2) -- (f);
    \draw[red, very thick] (f)+(0.2,0.2) -- (f);
    \draw[red, very thick] (f)+(-0.2,0.2) -- (f) node [above left] {$\tilde w_4$};
    \draw[red, very thick] (f)+(0.2,-0.2) -- (f);
    
    \fill[blue] (0,0) ++(18:4.2) circle[radius=3pt];
    \fill[blue] (0,0) ++(18:2.1) circle[radius=3pt] node [left] {$\tilde w_5$};
    \fill[blue] (0,0) ++(130:4.4) circle[radius=3pt];
    \fill[blue] (0,0) ++(130:1.9) circle[radius=3pt] node [below right] {$\tilde w_6$};
    \fill[blue] (0,0) ++(200:4.7) circle[radius=3pt];
    \fill[blue] (0,0) ++(200:1.6) circle[radius=3pt] node [right] {$\tilde w_7$};
    \fill[red] (0,0) ++(230:3.4) circle[radius=1pt] coordinate (g);
    \draw[red, very thick] (g)+(-0.2,-0.2) -- (g);
    \draw[red, very thick] (g)+(0.2,0.2) -- (g);
    \draw[red, very thick] (g)+(-0.2,0.2) -- (g);
    \draw[red, very thick] (g)+(0.2,-0.2) -- (g);
    \fill[red] (0,0) ++(230:2.8) circle[radius=1pt] coordinate (h);
    \draw[red, very thick] (h)+(-0.2,-0.2) -- (h);
    \draw[red, very thick] (h)+(0.2,0.2) node [right] {$\tilde w_8$} -- (h);
    \draw[red, very thick] (h)+(-0.2,0.2) -- (h);
    \draw[red, very thick] (h)+(0.2,-0.2) -- (h);
    \fill[blue] (0,0) ++(330:4.3) circle[radius=3pt];
    \fill[blue] (0,0) ++(330:2) circle[radius=3pt] node [above left] {$\tilde w_9$};
\end{tikzpicture}
\caption{Possible roots of $\tilde P$.}
\end{subfigure}
\caption{A visual representation of a plausible outcome of Root-MUSIC. (a) $P$ has four double roots on the unit circle $\{e^{i\theta_k}\}_{k=1}^4$ (in red marked by $\times$) and additional extraneous roots $\{(w_k,1/\overline w_k)\}_{k=5}^9$ (in blue marked by $\boldsymbol{\cdot}$). (b) A possible corresponding $\tilde P$, where the selected 4 roots in $\overline \D$ (in red marked by $\times$) are $\{\tilde w_k\}_{k=2,3,4,8}$. Notice how $\tilde w_8$ is selected since it is among the four roots closest to the unit circle even though it does not well-approximate any of the correct signal frequencies, while $\tilde w_1$ is not selected even though it serves as a better approximation.}
\label{fig:wrongroots}
\end{figure}

In the noiseless setting, the correct signal subspace $\bfU$ can be computed and the noiseless Root-MUSIC polynomial $P$ has exactly $s$ ``correct" roots on the unit circle, which are precisely $\{e^{i\theta_k}\}_{k=1}^s$. We call the remaining roots of $P$ ``extraneous", for which there could be up to $m-1-s$ many (mod reflection). Previous papers \cite{rao1989performance} and \cite{krim1992operator} provided asymptotic expectation formulas for the perturbation of correct roots under statistical assumptions. However, they are insufficient to guarantee the success of Root-MUSIC in the presence of noise. To see why,  \cref{fig:wrongroots} depicts a potential issue where a perturbed extraneous root of $P$ is closer to the unit circle than a perturbed correct root. Consequently, the Root-MUSIC algorithm would select a perturbed extraneous root whose argument (or phase, angle) does not approximate any of the correct signal frequencies, resulting in a failure of the estimation. This pathological example could result from $P$ having an extraneous root that is too close to the unit circle. 

One motivation of this paper is to study whether this behavior can actually occur under natural assumptions. Another motivation is to derive explicit and non-asymptotic expressions for the frequency error in terms of fundamental model parameters, such as the number of samples and noise level, with minimal assumptions. Letting $d(t,u)$ be the arc-length between $e^{it}$ and $e^{iu}$, we naturally define the 
$$
\text{frequency error}
:= \min_{\text{permutations } \sigma} \ \ \  \max_{k=1,\dots,s} \ \ \ d(\theta_k, \tilde \theta_{\sigma(k)}).
$$
This is also the $\ell^\infty$ matching distance between the correct and estimated frequencies---the permutation removes the significance of how the two sets are ordered.

The main obstacles lie in getting delicate control over perturbations of both the correct and extraneous roots of $P$. The most important technical assumption of this paper is a separation condition on the correct frequencies $\{\theta_k\}_{k=1}^s$. For a large enough absolute constant $\beta\geq 4$, we consider the situation where
\begin{equation}
    \label{eq:delta}
    \Delta := \min_{j\not=k} d(\theta_j,\theta_k)
    \geq \frac {2\pi \beta} m.
\end{equation}
This is a common and natural separation assumption found in other papers on spectral estimation such as \cite{candes2013super,duval2015exact,moitra2015matrixpencil,fannjiang2025optimality}. While certain algorithms can still succeed in the ``super-resolution regime" where $\Delta$ is significantly smaller than $\pi/m$, this setting is considerably different from the one considered here and is beyond the scope of this paper -- we refer the reader to \cite{li2021stable,batenkov2021super} and references therein for an overview of the mathematical theory of super-resolution. A second main assumption of this paper is that the noise is sufficiently small. This is captured in an abstract requirement that for an absolute $\epsilon>0$, 
$$
\text{subspace error} := \Big\|\bfU\bfU^*-\tilde \bfU\tilde \bfU^*\Big\|_2 \leq \epsilon.
$$
Note that the subspace error is the largest canonical angle between the subspaces $\bfU$ and $\tilde \bfU$. 

\subsection{Contributions}

\begin{table}[ht]
    \centering
    \begin{tabular}{|c|c|c|c|c|} \hline
         {\bf Section} &{\bf Signal model} &{\bf Noise class} &{\bf Subspace error} &{\bf Frequency error} \\\hline 
         \cref{sec:subspace1} &Single-snapshot &Deterministic &$\lesssim \|\bfeta\|_p/m^{1/p}$ &$\lesssim \|\bfeta\|_p/m^{1+1/p}$ \\ \hline 
         \cref{sec:subspace1} &Single-snapshot &i.i.d. $\calN(0,\sigma^2)$ &$\lesssim \sigma \sqrt{\log(m)/ m}$ &$\lesssim \sigma \sqrt{\log(m)}/ m^{3/2}$ \\ \hline 
         \cref{sec:subspace2} &Multi-snapshot &i.i.d. $\calN(0,\sigma^2)$ &$\lesssim \sigma/\sqrt{n\lambda_s(\bfSigma_\bfa)}$ &$\lesssim \sigma / (m \sqrt{n\lambda_s(\bfSigma_\bfa)})$ \\ \hline 
    \end{tabular}
    \caption{Summary of the main result applied to three noise models that are commonly considered in the literature where $\bfeta$ represents additive noise. For the first row, $p\in [1,\infty]$. For the second and third rows, the conclusions hold with probability $1-o(1)$ as $m\to\infty$. For the third row, $n$ is the number of independent snapshots and $\bfSigma_\bfa$ is the empirical covariance matrix of the amplitudes.}
    \label{table:mainresults}
\end{table}

The main results assume that $\Delta\geq 2\pi \beta/m$ for a large enough absolute $\beta\geq 4$ and that the subspace error is at most a small enough absolute constant. There are two distinct but equally important contributions of this paper (summarized in \cref{thm:step2}):
\begin{enumerate}
    \item 
    The first contribution of this paper shows that the situation depicted in \cref{fig:wrongroots} cannot occur by confirming that the $s$ roots (mod reflection) closest to the unit circle are perturbations of the correct roots under the above reasonable assumptions. 
    This is accomplished by showing a stronger and quantitative statement that $\tilde P$ has exactly $s$ roots (mod reflection) in an annulus of thickness on the order of $1/m$ around the unit circle. 
    \item 
    The second contribution of this paper is a unified, sharp, and non-asymptotic perturbation bound of the correct roots (which are also the ones selected by Root-MUSIC in view of the first contribution). There is an absolute constant $R>0$ such that Root-MUSIC finds roots $\{\tilde w_k\}_{k=1}^s$ with error
    \begin{equation}
        \label{eq:maininformal}
        \sup_{k=1,\dots,s} \left|e^{i\theta_k}-\tilde w_k \right|
        \leq \frac{R}{m} \left\|\bfU\bfU^*-\tilde \bfU\tilde \bfU^* \right\|_2.
    \end{equation}
    The general and abstract result \eqref{eq:maininformal} can be specialized to various specific spectral estimation settings and noise models which are summarized in \cref{table:mainresults}. 
\end{enumerate}

Even though this paper shows that the hypothetical situation depicted in \cref{fig:wrongroots} cannot occur under the main theorem's assumptions, the reader may wonder, ``why not avoid this issue altogether?". Indeed, one could discard any additional samples to reduce to the setting of $m=s+1$, in which case, $\tilde P$ (generally) has exactly $s$ roots (mod reflection) and no extraneous roots. While this would circumvent the issue depicted in \cref{fig:wrongroots}, assumption \eqref{eq:delta} relaxes as $m$ increases and the right side of \eqref{eq:maininformal} decays in $m$, so there is a clear advantage of utilizing additional measurements rather than discarding them. In other words, if one wants highly accurate estimates of the signal frequencies and utilize all available information, then extraneous roots must exist. However, the first contribution of the paper tells us that the extraneous roots can be easily avoided.

The additional $1/m$ factor in \eqref{eq:maininformal} has important practical implications. For example, under the setup of the first row of \cref{table:mainresults} and for $p=2$, the noise-to-signal ratio $\gamma^2$ satisfies $\gamma^2 \lesssim \|\bfeta\|_2^2/m$, see \cref{rem:NSR} for further details including the definition of signal-to-noise ratio $\gamma^2$. Consequently, the subspace error is $\lesssim \gamma$ and the frequency error is $\lesssim \gamma/m$. Of course, the frequency error goes to zero as the noise-to-signal ratio tends to zero; but perhaps more interestingly, the frequency error goes to zero as $m$ increases {\it even if the noise-to-signal ratio is fixed}. In other words, one can get increasingly better frequency estimation by utilizing additional samples without acquiring more accurate samples. 

The performance guarantees for Root-MUSIC in the first and second rows of \cref{table:mainresults}
match the minimax optimal rates (up to log factors), which were derived in \cite{stoica1989music,fannjiang2025optimality}. Hence, Root-MUSIC is an optimal spectral estimation algorithm over the class of signals with sufficiently separated frequencies and noise models considered in the first and second rows of \cref{table:mainresults}. 

The theory presented in this paper is verified by numerical simulations in \cref{sec:experiments}. Simulations agree with the prediction in the second row of \cref{table:mainresults}. Additionally, they show that the closest extraneous root is on the order of $1/m$ away from the annulus, which confirms the first contribution of this paper.

\subsection{Related work}

The first contribution of this paper, showing that Root-MUSIC does not select any perturbed extraneous root, is completely new, to our best knowledge. 

The second contribution of this paper, inequality \eqref{eq:maininformal} for the perturbation of correct roots, is also a new result. Past papers  \cite{rao1989performance} and \cite{krim1992operator} considered the multi-snapshot setting described in \cref{sec:subspace2} and produced asymptotic error formulas that depend on complicated and/or implicit quantities such as the noise subspace and the eigenvectors of the spatial covariance matrix. Consequently, it is not immediately clear how these bounds scale with the number of samples, sources, and different types of noise. In contrast, we obtain sharp and transparent inequalities which explicitly highlight how the error behaves in important model parameters as seen in \cref{table:mainresults}. A more technical and detailed comparison can be found in \cref{rem:techcomparison}. 

There exist modified Root-MUSIC algorithms, such as the one in \cite{wagner2021gridless}. This paper develops a modified Root-MUSIC algorithm to extend its applicability to non-uniform linear arrays. An open direction is to carry out perturbation analysis on this ``irregular Root-MUSIC" algorithm. We also refer the reader to \cite{wagner2021gridless} for numerical comparisons between Root-MUSIC, irregular Root-MUSIC, and other popular spectral estimation methods. The present paper is of a theoretical nature and does not advocate for one particular method over another. 

The strategy of bounding the frequency error in terms of the subspace error is a well-known template that appears in numerous other papers which analyze the performance of subspace methods \cite{rao1989performance,stoica1989music,krim1992operator,li2020super,li2022stability,yang2022nonasymptotic,gabet2025global,fannjiang2025optimality}. The steps required to prove such a result and the form in which it appears varies depending on which subspace method is being considered. As explained earlier, an important feature in \eqref{eq:maininformal} for Root-MUSIC is that the frequency error is at most $R/m$ times the subspace error, which has practical implications and implies optimality. This $1/m$ factor is not present in all spectral estimation theory. For instance, it is unknown if ESPRIT also achieves the general bound \eqref{eq:maininformal}. At the present moment, the only known algorithms that enjoy \eqref{eq:maininformal} are three MUSIC algorithms: classical MUSIC, Gradient-MUSIC (both proved in \cite{fannjiang2025optimality}) and Root-MUSIC. 

We refer the reader to \cite{li2021stable,batenkov2021super} for a detailed overview of further theoretical results, especially for the complementary setting of clustered frequencies.

\subsection{Outline}

\cref{sec:review} reviews the set up of Root-MUSIC, noting properties of the algorithm's steps as observations. \cref{sec:outline} reiterates the main theorem of the paper, its relevance to practical use, and the steps taken in the proof and their motivations (\cref{sec:outline1}). \cref{sec:subspace} describes necessary facts and assumptions regarding the noise, sine-theta distance, and noise-to-signal ratio for subspace estimation in practice for both the single-snapshot model (\cref{sec:subspace1}) and the multi-snapshot model (\cref{sec:subspace2}). It also includes the set up for the spectral estimation problem and a technical comparison to prior work. \cref{sec:experiments} details the numerical experiments that show the rates in simulation compared to the results of this paper. Finally, \cref{sec:proofs} contains all proofs for the main theorem and its steps as in Section 3. 

\subsection{Notation}

Conventional notation for inequality symbols are used. We write $x \lesssim_{a,b} y$, if there exists $c\in \R$ that depends only on $a,b$ such that  $x \leq cy$. In particular, $x\lesssim y$ means that the implicit $c>0$ is an absolute constant. The notation $x\asymp_{a,b} y$ means that $x\lesssim_{a,b} y$ and $y\lesssim_{a,b} x$. 

The open unit disk in the complex plane is $\D$. More generally, the open disk with center $w\in\C$ and radius $r>0$ is denoted $D(w,r)$. We use conventional set notation: for $A\subset\C$, let $\overline A$ be its closure and $\partial A$ its boundary. An (algebraic) polynomial of degree at most $n$ is a complex function $z\mapsto \sum_{k=0}^n c_k z^k$. A Laurent polynomial of degree at most $n$ is a function $z\mapsto \sum_{k=-n}^n c_k z^k$. We let $\T=\R/2\pi \Z$ be the periodic interval $[0,2\pi)$ so that $t\mapsto e^{it}$ is a bijection between $\T$ and $\partial\D$. Given a complex number $a \in \C$, let $\overline{a} \in \C$ denote its complex conjugate. 

An $m$ by $s$ matrix $\bfA$ with complex entries is denoted by 
$$\bfA = \left[ a_{j,k} \right]_{\substack{j=0,1,\dots,m-1 \\ k=1,2,\dots,s}} \in \C^{m\times s},$$ 
where $a_{j,k}$ is its entry in row $j+1$ and column $k$. A vector $\bfv$ in $\C^m$ is given by $\bfv = \left[ v_1, v_2, \dots , v_m \right]^\top$ where $\top$ denotes the transpose of the vector, and is equivalent to $\left[ v_j \right]_{j=1,\dots,m } \in \C^{m\times 1}$. $\bfA^*$ is the complex conjugate transpose of a complex matrix $\bfA$ (or a vector $\bfv$) as before, so that 
$$\bfA^* = \left[ \overline{a_{k,j}} \right]_{\substack{k=1,2,\dots,s \\ j=0,1,\dots,m-1}}.$$ 

Concerning the norms of vectors and matrices, $\| \bfv \|_2$ denotes the $\ell^2$ norm of a vector $\bfv$ and $\| \bfA \|_2$ for a matrix $\bfA$ denotes the spectral norm, given by 
$$\| \bfA \|_2 = \sqrt{\lambda_{\max}(\bfA^*\bfA)} = \sigma_{\max}(\bfA),$$
where $\lambda_{\max}(\bfA^*\bfA)$ and $\sigma_{\max}(\bfA)$ are the maximum eigenvalue of $\bfA^*\bfA$ and the maximum singular value of $\bfA$ respectively. Additionally, the $L^{\infty}$ norm of a function $f$ defined on $\T$ is denoted 
$$\|f\|_{L^{\infty}(\T)} := \sup_{t\in\T} |f(t)|.$$

The range of a matrix $\bfA$ is the span of the column vectors which make up $\bfA$: 
$$
\text{range}(\bfA) = \text{span}(\{\bfa_k\}_{k=1,2,\dots,s})\subset \C^m, 
$$
where $\bfa_k = \left[ a_{0,k}, a_{1,k}, \dots , a_{m-1,k} \right]^\top$ for all $k=1,2,\dots,s$. The trace of a matrix is denoted $\tr(\bfA)$. 

We write $X \sim Y$ if $X$ and $Y$ are random variables with the same distribution. We say a mean-zero random vector $\bfu$ in $\R^m$ is $K$-subgaussian if 
$$
\| \langle \bfu,\bfv\rangle\|_{\psi_2}
\leq K \sqrt{\E |\langle \bfu,\bfv\rangle|^2} \forallspace \bfv\in \R^m,
$$
where the subgaussian $\psi_2$-norm of a random variable $U$ is
$$
\|U\|_{\psi_2}
:= \inf \left\{ t>0 \colon \E \exp(|U|^2/t) \leq 2 \right\}.
$$
A mean-zero random vector $\bfu$ in $\C^m$ is $K$-subgaussian if its real and imaginary parts are independent and are both $K$-subgaussian.

\section{Review of Root-MUSIC}
\label{sec:review} 

This section reviews the Root-MUSIC algorithm, which was first developed in \cite{barabell1983improving} and later analyzed in \cite{rao1989performance,krim1992operator}. Several known facts related to the classical and Root-MUSIC algorithms are stated as {\bf Observation}, which are not reproved unless  mentioned specifically. Compared to \cite{rao1989performance}, our exposition is more streamlined because \cite{rao1989performance} derives a special function (which will be denoted $Q$ in our notation) from first principles, whereas we take a shortcut of starting with the correct definition of $Q$.

{\bf First step. Subspace estimation.} The first step of MUSIC algorithms and other subspace methods is to estimate the subspace on which the correct signal parameters $\{\theta_k\}_{k=1}^s$ lie on. To this end, define the Fourier (or Vandermonde) matrix with $m$ rows associated to $\{\theta_k\}_{k=1}^s$ as
$$
\bfPhi:=\bfPhi(m,\{\theta_k\}_{k=1}^s) 
:= \left[ e^{i j \theta_k} \right]_{\substack{j=0,1,\dots,m-1 \\ k=1,2,\dots,s}}
\in \C^{m\times s}
$$
Here, the columns of $\bfPhi$ have $\ell^2$ norm $\sqrt m$. The correct signal subspace in $\C^m$ associated with $\{\theta_k\}_{k=1}^s$ is precisely
$$
\bfU:=\bfU(m,\{\theta_k\}_{k=1}^s)=\text{range}(\bfPhi(m,\{\theta_k\}_{k=1}^s))\subset \C^m. 
$$
Although we will not need the following fact, we mention that there is a bijection between all subspaces in $\C^m$ defined this way and all possible subsets of cardinality $s$ contained in $\T$, see \cite[Lemma 5.3]{fannjiang2025optimality}. We will frequently drop dependence on $m$ and $\{\theta_k\}_{k=1}^s$ which will be understood from context, and the symbol $\bfU$ is reserved for denoting the subspace associated to $\{\theta_k\}_{k=1}^s$. Let $\bfU_\perp:=\bfU_\perp(m,\{\theta_k\}_{k=1}^s)$ be the orthogonal complement of $\bfU$. 

Throughout, we slightly abuse notation and let $\bfU$ (resp., $\bfU_\perp$) also denote a matrix whose columns form an orthonormal basis for the subspace $\bfU$ (resp., $\bfU_\perp$). Then $\bfU\bfU^*$ is the projection onto its range. As we will see, choice of bases for the ranges of $\bfU$ and $\bfU_\perp$ are completely inconsequential for the MUSIC-type algorithms.

Subspace estimation is about computing an approximate subspace $\tilde \bfU$ of $\bfU$ from noisy observations. The theory developed in this paper will be completely agnostic to the method in which an $\tilde \bfU$ is computed because the main theorem (\cref{thm:step2}) will be applicable to any $\tilde \bfU$ for which the subspace error is small enough. Since the subspace error will appear frequently in this paper, we denote it as
\begin{equation}
    \label{eq:subspaceerror}
    \rho := \left\|\bfU\bfU^*-\tilde \bfU\tilde \bfU^* \right\|_2.
\end{equation}
This quantity is the sine-theta distance between $\bfU$ and $\tilde \bfU$, which makes it a natural metric. The theory of subspace estimation is fairy well-understood at this point and is not the main focus of this paper. For completeness, we provide two concrete spectral estimation models in \cref{sec:subspace}, the single- and multi- snapshot cases. They rely on computing $\tilde\bfU$ through a truncated singular value decomposition of a Hankel or covariance matrix. 

{\bf Second step. Root-MUSIC polynomial and root-finding.} Now we are in position to explain the general mechanisms of the Root-MUSIC algorithm. Define the ``steering vector" $\bfphi\colon \C\to \C^m$ by
$$
\bfphi(z)
=\frac 1 {\sqrt m} \begin{bmatrix} 1, &z, &z^2,  &\cdots, &z^{m-1} \end{bmatrix}^\top.
$$
Here we have chosen to normalize $\bfphi$ by $1/\sqrt m$ so that $\|\bfphi(e^{it})\|_2=1$ for all $t\in \R$. Since $\bfphi$ is a vector valued differentiable function of $z$, we have
$$
\bfphi'(z)=\frac 1 {\sqrt m} \begin{bmatrix} 0, &1, &2z,  &\cdots, &(m-1)z^{m-2} \end{bmatrix}^\top.
$$
Consider the complex function $Q\colon \C\setminus \{0\}\to \C$ defined as 
\begin{equation}
    Q(z)
    :=\bfphi(1/\overline z)^* \bfU_\perp \bfU_\perp^* \bfphi(z). \label{eq:qdef1}
\end{equation}
Using that $\bfU_\perp \bfU_\perp^*+\bfU \bfU^*=\bfI$ and that $\bfphi(1/\overline z)^*\bfphi(z)=1$ for all $z\in \C$, we obtain the alternative formula
\begin{equation}
    Q(z)
    =1-\bfphi(1/\overline z)^* \, \bfU \bfU^* \bfphi(z). \label{eq:qdef2}
\end{equation}

There is a close relationship between the function constructed in the classical MUSIC algorithm \cite{schmidt1986multiple} and $Q$. If $Q$ is restricted to the unit complex circle and viewed as a $2\pi$ periodic function, then we get the function $q\colon \T\to [0,1]$ defined as 
$$
q(t):= Q(e^{it}). 
$$
The classical MUSIC algorithm seeks to find the $s$ smallest local minima of $q$, typically done with a brute force grid search. The MUSIC algorithm relies on the following equivalence \cite{schmidt1986multiple} whose proof we omit.

\begin{observation}
	\label{obs:1}
	Suppose $m\geq s+1$. Fix any subset $\{\theta_k\}_{k=1}^s$ and let $\bfU$ be its associated subspace. For any $t\in \T$, it holds that $q(t)=0$ if and only if $\bfphi(e^{it})\in \bfU$ if and only if $t\in \{\theta_k\}_{k=1}^s$. Moreover, each $\theta_k$ is a double zero of $q$.
\end{observation}

While the last part of this observation asserts that $q''(\theta_k)>0$ for each correct zero $\theta_k$, there is a quantitative lower bound in \cite{fannjiang2025optimality} which will be stated in \cref{lem:q''}. 

A computation shows that $Q$ is a Laurent polynomial of degree at most $m-1$ and can be written as
\begin{equation}
    \label{eq:defQ}
    Q(z):=\sum_{j=-m+1}^{m-1} q_j z^j, 
\end{equation}
for some coefficients $\{q_j\}_{j=-m+1}^{m-1}$, which are also the Fourier coefficients of $q$. 

Importantly, $Q$ satisfies a certain symmetry property that requires introducing some terminology. For each $z\not=0$, define {\it reflection} through the circle by
$$
z\to \frac 1{(\overline z)}.
$$ 
This fixes the unit circle and maps $\D\setminus\{0\}$ to the complement of $\overline \D$. It also does not change the argument of a complex number and only inverts its modulus. For this reason, we have called it reflection, which is consistent with other convention such as \cite[Chapter I.3, Exercise 2]{lang1999complex}. 

The zeros of $Q$ come in pairs of the form $(w,1/\overline w)$. This was first proved in \cite{rao1989performance}, but due to the importance of this property for Root-MUSIC, we reprove it for the reader's convenience in \cref{proof:observations}. 

\begin{observation}
	\label{obs:3}
	Suppose $m\geq s+1$. Then $Q$ has exactly $2s$ zeros on $\partial\D$, which are precisely $\{e^{i\theta_k}\}_{k=1}^s$, each with multiplicity two. Additionally $Q$ has at most $2m-2$ complex zeros, and they come in pairs with the following reflection symmetry: $w$ is a zero of $Q$ if and only if $1/\overline w$ is a zero. Finally, there is an algebraic polynomial $B$ of degree at most $m-1$ whose roots are all in $\overline \D$ such that
	\begin{equation}
		\label{eq:FR}
		Q(z)=B(z)\overline{B(1/\overline z)} \forallspace z\in \C.
	\end{equation}
    Consequently, we have the identity 
    $$
    Q(z)=\overline{Q(1/\overline z)} \forallspace z\in \C\setminus \{0\}.
    $$
\end{observation}

As a consequence of Observation \ref{obs:3}, we enumerate the $n\leq m-1$ roots of the polynomial $B$ by 
\begin{equation}
	\label{eq:roots1}
	w_1=e^{i \theta_1}, \, w_2 = e^{i\theta_2}, \,  \dots, w_s=e^{i \theta_s}, \, w_{s+1}, \, \dots, \, w_n.
\end{equation}
Here, the first $s$ roots lie on $\partial\D$, while the remaining $n-s$ roots have to be in $\D$ due to \cref{obs:1}. We use this convention throughout the paper. 

In the presence of noise, suppose we obtain a $\tilde\bfU$ that approximates $\bfU$, and let $\tilde \bfU_\perp$ be a matrix whose columns form an orthonormal basis for the orthogonal complement of the range of $\tilde\bfU$. Again, see \cref{sec:subspace} for examples. Using $\tilde \bfU$ instead of $\bfU$ in formulas \eqref{eq:qdef1} and \eqref{eq:qdef2} yields a perturbed function,
\begin{equation}
    \tilde Q(z)
    =\bfphi(1/\overline z)^* \tilde \bfU_\perp \tilde \bfU_\perp^* \bfphi(z)
    =1 -\bfphi(1/\overline z)^* \tilde \bfU \tilde \bfU^* \bfphi(z). \label{eq:defQtilde1}
\end{equation}
This function is also a Laurent polynomial of degree at most $m-1$ and can be written as
\begin{equation}
    \label{eq:defQtilde2}
    \tilde Q(z):=\sum_{j=-m+1}^{m-1} \tilde q_j z^j,
\end{equation}
for some coefficients $\{\tilde q_k\}_{k=-m+1}^{m-1}$. The Laurent polynomial $\tilde Q$ also enjoys symmetry properties, which we reprove \commentout{it }for the reader's convenience in \cref{proof:observations}.

\begin{observation}
	\label{obs:4}
	Suppose $m\geq s+1$. Then $\tilde Q$ has an even number of complex zeros and they come in pairs with the following reflection symmetry: $w$ is a zero if and only if $1/\overline{w}$ is a zero. We have the identity 
    $$
    \tilde Q(z)=\overline{\tilde Q(1/\overline z)} \forallspace z\in \C\setminus \{0\}.
    $$
\end{observation}

Unlike $Q$, the Laurent polynomial $\tilde Q$ does not necessarily have any zeros on the unit circle and it is not immediate that $\tilde Q$ has any zeros whatsoever (e.g., it is not immediately clear that $\tilde Q$ is non-constant). 

Instead of working with the Laurent polynomials $Q$ and $\tilde Q$, it is more convenient to work with the {\it Root-MUSIC polynomial} $P$ and its noisy version $\tilde P$, which are defined as
\begin{equation}
    \label{eq:Pdef} 
    P(z)= z^{m-1} Q(z)\andspace 
    \tilde P(z)= z^{m-1} \tilde Q(z).
\end{equation}
Both are polynomials of degree at most $2m-2$ in view of equations \eqref{eq:defQ} and \eqref{eq:defQtilde2}. It is also clear that $Q$ and $P$ (resp. $\tilde Q$ and $\tilde P$) have the same zeros in $\C\setminus\{0\}$, though $P$ and $\tilde P$ may vanish at zero. We will eventually show in \cref{thm:step1} that $\tilde Q$ and $\tilde P$ have at least $2s$ many zeros and quantify their locations.  

Root-MUSIC computes the roots of $\tilde P$ through a root finding algorithm (such as Newton's method), selects the $s$ roots mod reflection that are the closest to the circle (if they exist), and their arguments are the estimated signal frequencies. 

We need to be careful about the definition of ``closest to the circle" since it should not be defined in the usual sense. The roots of $\tilde P$ necessarily come in pairs satisfying the reflection symmetry provided in \cref{obs:4}. Since this reflection does not change the argument of a complex number, it is not important if a root $w$ or its reflection $1/\overline w$ is found by the root finding procedure, despite them having different distances to the circle. For this reason, it is natural to measure the distance between $w\in \C$ to the circle by
\begin{equation}
    \min\left\{ \left||w|-1 \right|, \, \left| |w|^{-1}-1 \right| \right\},
    \label{eq:defdistance}
\end{equation}
which is clearly invariant under reflection of $w$. Hence, the Root-MUSIC algorithm selects the $s$ distinct roots modulo reflection symmetry that are closest to the circle in the sense of \eqref{eq:defdistance}. Throughout, the $s$ roots of $\tilde P$ in 
$\overline \D$ that are closest to the circle are denoted
$$
\tilde w_1,\dots,\tilde w_s\in\overline \D.
$$ 
This is purely a matter of conventional convenience, since in practice, $1/\overline {\tilde w_k}$ may be found by the root finding procedure instead of $\tilde w_k$, which has no impact on the final output of the algorithm since $\tilde w_k$ and $1/\overline{\tilde w_k}$ have the same argument. 

As mentioned in the introduction and depicted in \cref{fig:wrongroots}, there is a possibility that the Root-MUSIC algorithm selects a perturbation of an extraneous root of $P$, which would lead to failure. Although the notation is suggestive, we defined $\tilde w_k$ as a root selected by the algorithm and {\it it is not defined as a perturbation of the correct root $w_k=e^{i\theta_k}$}. Under the assumptions of \cref{thm:step2}, we will show that $\tilde w_k$ is in fact a perturbation of the correct root $w_k$.

\section{Main result and outline of proof}
\label{sec:outline} 

The following theorem is the main result of this paper and is proved in \cref{proof:step2}.

\begin{theorem}[Main theorem for Root-MUSIC] \label{thm:step2}
    There are absolute constants $R,\epsilon,\tau,\beta>0$ and $m_0\geq 100$ such that the following hold. For any $m\geq m_0$, $\{\theta_k\}_{k=1}^s\subset\T$ such that $\Delta(\{\theta_k\}_{k=1}^s)\geq 2\pi\beta/m$, and subspace $\tilde\bfU$ such that $\|\bfU \bfU^* -\tilde\bfU \tilde \bfU^*\|_2 \leq \epsilon$, the Root-MUSIC polynomial $\tilde P$ enjoys the following properties. 
    \begin{enumerate}
        \item 
        It has exactly $2s$ roots, denoted $\{\tilde w_k, 1/\overline{\tilde w_k}\}_{k=1}^s$, in the annulus
        $$
        \left\{z\in\C\colon 1-\frac \tau m \leq |z|\leq \left( 1- \frac \tau m\right)^{-1} \right\}.
        $$
        \item 
        For each $k\in \{1,\dots,s\}$, 
        \begin{equation}
    		\max\left\{ \left| e^{i\theta_k} -\tilde w_k \right|, \left| e^{i\theta_k} - 1/\overline{\tilde w_k}\right| \right\}\leq \frac {R}{m}\left\|\bfU\bfU^*- \tilde \bfU \tilde\bfU^* \right\|_2.  \label{eq:mainformal}
    	\end{equation}
    \end{enumerate}
\end{theorem}

Part (1) of this theorem gives us that $\tilde P$ has at least $s$ roots (mod reflection), and that the $s$ roots produced by the Root-MUSIC algorithm (which selects the $s$ roots closest to the unit circle) are the only roots in an explicit annulus around the unit circle. 
Part (2) confirms that those roots are close to $\{e^{i\theta_k}\}_{k=1}^s$ as desired. It also says that as the subspace error decreases or the number of samples increases, the distances of the selected roots of $\tilde P$ to $\{e^{i\theta_k}\}_{k=1}^s$ decrease as well. 

Control of the subspace error is possible without knowledge of $\bfU$ due to the relationship between the subspace error and noise $\bfeta$ given in \cref{table:mainresults} and \cref{sec:subspace}. Thus, if the $\ell^p$ norm of $\bfeta$ is sufficiently small (for deterministic case) or if $m$ or $n$ is large enough (for stochastic cases), then the subspace error will be bounded by a sufficiently small constant, which enables the use of \cref{thm:step2}. Of course, if $\bfeta$ is too large, then subspace estimation is an impossible task, which is why the assumption that the subspace error is at most a small enough absolute constant is necessary in general. 

The bound in \eqref{eq:mainformal} cannot be predicted from heuristic arguments and standard wisdom. Since $\|\bfphi(w)\|_2=1$ whenever $w\in\partial\D$, 
$$
\sup_{w\in \partial\D} |P(w)-\tilde P(w)|
\leq \left\|\bfU\bfU^*- \tilde \bfU \tilde\bfU^* \right\|_2.
$$
That is, the perturbation of $P$ on the circle scales linearly with the subspace error with a constant that does not depend on $m$ (this is generally sharp and cannot be improved). To make matters worse, perturbation theory says that double roots are generally more sensitive to perturbations than simple roots \cite{simmonds1998first} . To give an example, the quadratic polynomial $z^2$ has a double root at $z=0$, yet a $\epsilon>0$ perturbation of $z^2$ may result in the perturbed polynomial $z^2-\epsilon$, which has a roots $\pm \sqrt \epsilon$. Each $e^{i\theta_k}$ is a double root of $P$, so general theory cannot produce anything better than
$$
\max\left\{ \left| e^{i\theta_k} -\tilde w_k \right|, \left| e^{i\theta_k} - 1/\overline{\tilde w_k}\right| \right\}\lesssim \left\|\bfU\bfU^*- \tilde \bfU \tilde\bfU^* \right\|_2^{1/2}. 
$$
Here, the implicit constant at best does not grow in $m$ and the square root comes from $e^{i\theta_k}$ being a double root. This bound is obviously much worse than \eqref{eq:mainformal} which scales linearly in the subspace error and has an additional $1/m$ factor. 

An outline of the proof of \cref{thm:step2} is given in \cref{sec:outline}. There, we explain how $P$ has special properties and how $P-\tilde P$ has hidden cancellations in an appropriate neighborhood of each $e^{i\theta_k}$, which enable us to prove \cref{thm:step2} through Rouche's theorem and parameter tuning arguments. It is precisely these special properties which enables $P$ and $\tilde P$ to avoid bad behavior that are present for general polynomials with double roots.

\subsection{Outline of Proof of \cref{thm:step2}}
\label{sec:outline1} 

To relate the zeros of $P$ and $\tilde P$, we will use Rouche's theorem, \cite[page 153]{ahlfors1979complex}. 

\begin{theorem}[Rouche] \label{thm:rouche}
	Let $\gamma$ be homologous to zero in a set $U\subset \C$ such that for each $z\not\in \gamma$, the winding number $n(\gamma,z)$ of $\gamma$ with respect to $z$ is either zero or one. If $f$ and $g$ are analytic in $U$ and $|f(z)-g(z)|<|f(z)|$ for all $z\in \gamma$, then $f$ and $g$ have the same number of zeros in the set enclosed by $\gamma$. 
\end{theorem}

For example, a simple closed curve $\gamma$ in the complex plane is homologous to zero. As mentioned in the introduction, it is convenient to think of $\tilde \bfU$ as fixed and phrase everything in terms of the subspace error,
$$
\rho
:= \left\|\bfU\bfU^*- \tilde \bfU \tilde\bfU^* \right\|_2.
$$

{\bf Step 1. Control over the correct roots.} Recall our notation that $w_1=e^{i\theta_1},\dots,w_s=e^{i\theta_s}$ are the roots of $P$ on $\partial \D$, where each one has multiplicity two. For appropriately chosen $R>0$ and an $\alpha > R\rho$, we will apply Rouche's theorem to $P$ and $\tilde P$, and use circles 
$$
C_k:=\partial D\left(w_k, \frac{R\rho} m\right), \quad \tilde C_k:=\partial D\left(w_k, \frac{\alpha} m\right). 
$$
Assuming this can be done,  Rouche's theorem implies that $\tilde P$ has 
\begin{itemize}
    \item 
    exactly two roots of the form $z_k$ and $1/\overline{z_k}$ in the disk enclosed by $C_k$.
    \item 
    exactly two roots in the disk enclosed by $\tilde C_k$. 
\end{itemize}
Together, they imply that $z_k$ and $1/\overline{z_k}$ are both inside the disk enclosed by $C_k$. In order to use Rouche's theorem, we need to to lower bound $|P(z)|$ (step 1a) and upper bound $|P(z)-\tilde P(z)|$ (step 1b), when $z\in C_k, \tilde C_k$. 

{\bf Step 1a. Lower bounds for $|P|$ near $w_k$.} Since $P$ has a double root at $w_k$, it is locally a quadratic function. Whenever $z$ is near $w_k$, we expect that  
$$
P(z)= \frac 12 P''(w_k)(z-w_k)^2 \left( 1+O(z-w_k)\right) \quad \text{as} \quad z\to w_k.
$$
To derive a quantitative bound for how small $|z-w_k|$ needs to be, the key step is to show that 
$
|P''(w_k)|\geq c_0 m^2
$
for an absolute constant $c_0>0$. This step leverages a result from prior analysis of MUSIC in \cite{fannjiang2025optimality}. This will allow us to establish that
\begin{equation}
    \label{eq:Plower0}
    |P(z)|
    \gtrsim m^2 |z-w_k|^2
    = \begin{cases}
        R^2 \rho^2 &\forspace z\in C_k, \\
         \alpha^2 &\forspace z\in \tilde C_k. 
    \end{cases}
\end{equation}
 
{\bf Step 1b. Upper bounds for $|P-\tilde P|$ near $w_k$.} Let us first examine the simpler matter of controlling $|P-\tilde P|$ on the unit circle. A naive estimate for this is 
$$
\|P-\tilde P\|_{L^\infty(\partial \D)}
\leq \rho.
$$
This  is far too weak to invoke Rouche's theorem for $C_k$ because in view of \eqref{eq:Plower0}, we will not be able to establish that $|P(z)-\tilde P(z)| < |P(z)|$ for $z\in C_k$ and all small $\rho$. So at this point, it seems impossible to use Rouche's theorem without imposing a stronger noise condition or making the radius of $C_k$ larger (hence would yield a worse perturbation bound). 

However, there is some additional structure in $P$ that leads to some hidden cancellations. We establish a substantially stronger local estimate,
\begin{equation}
    \label{eq:Pupper0}
    |P(z)-\tilde P(z)|\lesssim \begin{cases}
        R\rho^2 &\forspace z\in C_k,\\
        \alpha &\forspace z\in \tilde C_k.
    \end{cases}
\end{equation}
Importantly, the upper bound \eqref{eq:Pupper0} is linear in $R, \alpha$ while the lower bound \eqref{eq:Plower0} is quadratic in $R, \alpha$, and both have the same scaling in $\rho$. By making $R, \alpha$ large enough absolute constants, we can apply Rouche's theorem. The case for $\tilde C_k$ is proven as \cref{lem:coarseroots}. 

The following theorem summarizes Step 1 and is proved in \cref{proof:step1}.

\begin{theorem}[Proximity of the Correct Roots]
	\label{thm:step1}
	There are absolute constants $R,\epsilon_1>0$ such that the following hold. Let $m\geq 100$, $\{\theta_k\}_{k=1}^s\subset\T$ such that $\Delta(\{\theta_k\}_{k=1}^s)\geq 8\pi/m$, and $\tilde\bfU$ such that 
    $$
    \rho:=\|\bfU\bfU^*-\tilde\bfU\tilde\bfU^*\|_2\leq \epsilon_1.
    $$
    The Root-MUSIC polynomial $\tilde P$ has exactly two roots with reflection symmetry that are contained in the disk $D(w_k,R\rho/m)$ for each $k\in \{1,\dots,s\}$. 
\end{theorem}

{\bf Step 2. Control over the extraneous roots.} The first step is insufficient to complete the main objective of this paper because the remaining roots of $\tilde P$ could be closer to $\partial\D$ than the $s$ roots (mod reflection) described in the theorem, which could result in the Root-MUSIC algorithm selecting a perturbed extraneous root. To show that this does not happen, we will prove a stronger assertion that $\tilde P$ has no other roots in an appropriate annulus containing $\partial \D$.

It is helpful to work with wedges that respect the natural reflection symmetry of the Root-MUSIC polynomial. For a $w\in \partial\D$, $r \in (0,1)$, and $\phi\in [0,\pi]$, define the closed wedge
\begin{equation}
	W(w,r,\phi)
	:=\left\{z\in \C\colon r \leq |z|\leq r^{-1}, \, |\arg(z)-\arg(w)|\leq \phi \right\}. \label{eq:defwedge}
\end{equation}
Note that $W(w,r,\phi)$ is invariant under reflection across the circle. Also, $W(w,r,\pi)$ is an annulus with radial thickness asymptotically equal to $2(1-r)$ as $r\to 1$.

The second assertion of \cref{thm:step1} implies that $\tilde P$ has no other roots besides $z_k$ and $1/\overline{z_k}$ in a $\alpha/m$ neighborhood of each $w_k$. We will show that for some absolute constants $\tau_0,\alpha_0>0$, the polynomial $\tilde P$ has no other roots aside from $z_k$ and $1/\overline{z_k}$ in the wedge 
$$
W_k:=W\left(w_k, 1-\frac {\tau_0} m, \frac {\alpha_0} m\right). 
$$

Next, we will show that $\tilde P$ has no other zeros aside from $\{ z_k\}_{k=1}^s$ and their reflections in the annulus
\begin{equation}
    \label{eq:naturalannulus}
    \left\{z\in\C\colon 1-\frac {\tau_0} m \leq |z|\leq \left( 1- \frac {\tau_0} m\right)^{-1} \right\}. 
\end{equation}
This will be accomplished by first deriving a lower bound for $|P(z)|$ when $|\arg(z)-\arg(w_k)|\geq \alpha_0/m$ for each $w_k$. Through a radial perturbation argument, we will show that if $z$ lies in the annulus \eqref{eq:naturalannulus} and $|\arg(z)-\arg(w_k)|\geq \alpha_0/m$ for all $k$, then $|\tilde P(z)|>0$. This will show that $\tilde P$ has precisely $s$ roots (mod reflection) in an annulus of thickness $\tau_0/m$. 

Consequently, the Root-MUSIC algorithm will always select either $z_k$ or $1/\overline{z_k}$ for $k=1,\dots,s$. As a result, we can say that $z_k=\tilde w_k$.

\section{Subspace estimation for two models}
\label{sec:subspace}

\subsection{Single-snapshot model}
\label{sec:subspace1}

Consider a finite Fourier series $h\colon\R\to\C$ defined as
\begin{equation}
	\label{eq:hdef}
	h(t):= \sum_{k=1}^s a_k e^{i\theta_k t}. 
\end{equation}
Without the knowledge of its ``amplitudes" $\bfa=\{a_k\}_{k=1}^s\subset\C\setminus \{0\}$ and its ``frequencies" $\{\theta_k\}_{k=1}^s\subset \T$, the single-snapshot spectral estimation problem is to accurately approximate $\{\theta_k\}_{k=1}^s$ from perturbed samples of $h$. Suppose we collect $2m-1$ consecutive samples,
$$
\tilde \bfy = \begin{bmatrix} h(j) \end{bmatrix}_{j=0,1,\dots,2m-2}+\bfeta, 
$$
where $\bfeta\in \C^{2m-1}$ represents some unknown perturbation which could be deterministic/adversarial or stochastic. By rescaling, we can assume without loss of generality that $|a_k|\geq 1$ for each $k$. 

To compute a $\tilde\bfU$, we form a Hankel matrix
$$
H(\tilde \bfy)
=\begin{bmatrix}
    \tilde y_0 &\tilde y_1 &\cdots &\tilde y_{m-1} \\
    \tilde y_1 &\iddots &\iddots &\tilde y_{m} \\
    \vdots &\iddots &\iddots &\vdots \\
    \tilde y_{m-1} &\tilde y_m &\cdots &\tilde y_{2m-2} 
\end{bmatrix} \in \C^{m\times m}. 
$$
Under appropriate conditions,
\begin{equation}
    \label{eq:defUtilde1}
    \tilde \bfU:= \text{top $s$-left singular vectors of } H(\bfy).
\end{equation}
The following result was implicitly proved in the proof of \cite[Theorems 4.1 and 4.2]{fannjiang2025optimality}. It is essentially a combination of Wedin's theorem with quantitative control over the smallest singular value of Fourier matrix \cite{aubel2019vandermonde} when $m\Delta \gtrsim 1$. 

\begin{proposition}
    There are absolute constants $C,c>0$  such that the following hold. Let $m\geq 100$ and $\{\theta_k\}_{k=1}^s\subset\T$ such that $\Delta(\{\theta_k\}_{k=1}^s)\geq 8\pi/m$. 
    \begin{enumerate}[(a)]
        \item 
        Let $p\in [1,\infty]$ and $\bfeta\in \C^{2m-1}$ with $\|\bfeta\|_p\leq c m^{1/p}$. Then $\tilde \bfU$ in \eqref{eq:defUtilde1} is well-defined and 
        $$
        \left\|\bfU \bfU^* - \tilde \bfU\tilde \bfU^*\right\|_2
        \leq \frac{C \|\bfeta\|_p}{m^{1/p}}. 
        $$
        \item 
        Let $\bfeta\sim \calN(\bfzero, \bfSigma)$ where $\bfSigma$ is diagonal. For any $t>0$, with probability at least $1-2m^{1-t^2}-2m\exp(-c m^2/\tr(\bfSigma))$, the subspace $\tilde\bfU$ in \eqref{eq:defUtilde1} is well-defined and 
        $$
        \left\|\bfU \bfU^* - \tilde \bfU\tilde \bfU^* \right\|_2
        \leq \frac{C t \sqrt{\tr(\bfSigma) \log(m)}}{m}. 
        $$
    \end{enumerate}
    \label{prop:subspace1}
\end{proposition}

In the first case of deterministic noise, by assuming $\|\bfeta\|_p \leq c m^{1/p}$ for an absolute $c>0$, the subspace error can be made small by making $c$ small enough. Combining this with \cref{thm:step2} yields the conclusions in the first row of \cref{table:mainresults}. 

In the second case of stochastic noise, a standard example is $\bfSigma=\sigma^2 \bfI$, in which case $\tr(\bfSigma)=\sigma^2 (2m-1)$ and the subspace error is $\lesssim \sigma \sqrt{\log(m)/m}$. By assuming $m$ is large enough, the subspace error can be made arbitrarily small. Combining this with \cref{thm:step2} yields the conclusions in the second row of \cref{table:mainresults}. 

In both cases, if the subspace error is small enough, then the correct value of $s$ can be determined by examining the singular values of $H(\tilde \bfy)$ as shown in \cite[Lemma 5.6]{fannjiang2025optimality}. 

\begin{remark}[Statement in terms of noise-to-signal ratio]
\label{rem:NSR} 
    The standard definition of noise-to-signal ratio 
    $\gamma^2$ is the ratio of the noise energy $\|\bfeta\|_2^2$ versus signal energy $\|\bfy\|_2^2$, where $\bfy$ denotes the noiseless signal and is $\tilde \bfy$ when $\bfeta=\bfzero$. Recall that $\sigma_1(\bfPhi)\asymp \sigma_s(\bfPhi)\asymp\sqrt m$ when $m\Delta \geq 8\pi$ due to \cite{aubel2019vandermonde}. Also recall that $\min |a_k|\geq 1$ and that $\bfy=\bfPhi \bfa$. Putting these together, we see that
    $$
    \frac{\|\bfeta\|_2^2}{m \|\bfa\|_2^2}
    \lesssim \frac{\|\bfeta\|_2^2}{\|\bfy\|_2^2}
    \lesssim \frac{\|\bfeta\|_2^2}{m}.
    $$
    \cref{prop:subspace1} says that if $\gamma\lesssim 1$, then the subspace error is $\lesssim \gamma$, and \cref{thm:step2} says that Root-MUSIC estimates the signal frequencies with error $\lesssim \gamma/m$.
\end{remark}

\subsection{Multi-snapshot model}
\label{sec:subspace2}

For $\ell=1,\dots,n$, consider a finite Fourier series $h_\ell\colon\R \to\C$ defined as
\begin{equation}
	\label{eq:hdef2}
	h_\ell(t):= \sum_{k=1}^s a_{k,\ell} \, e^{i\theta_k t}. 
\end{equation}
For each $\ell$, the amplitude vector $\bfa_\ell:=\{ a_{\ell,k}\}_{k=1}^s$ varies whereas the frequencies $\{\theta_k\}_{k=1}^s$ do not. Without the knowledge of  $\bfa_1,\dots,\bfa_n$ and $\{\theta_k\}_{k=1}^s\subset \T$, the multi-snapshot spectral estimation problem is to accurately approximate $\{\theta_k\}_{k=1}^s$ from perturbed samples of $h_1,\dots,h_n$. This is sometimes also referred to as collecting multiple measurement vectors. 

More specifically, suppose for each $\ell$, we obtain noisy samples of $h_\ell$ at $t=0,1,\dots,m-1$. The $\ell$-th snapshot is the vector 
$$
\tilde \bfy_\ell = \begin{bmatrix} h_\ell(j) \end{bmatrix}_{j=0,1,\dots,m-1} + \bfeta_\ell,
$$
where $\bfeta_\ell\in \C^m$ represents noise for the $\ell$-th snapshot. We can collect all of these measurements into matrices,
\begin{align*}
    \tilde\bfY
&:=\begin{bmatrix}
    \tilde \bfy_1 &\cdots &\tilde \bfy_n
\end{bmatrix}, \\
\bfA&:= \begin{bmatrix}
        \bfa_1 &\cdots &\bfa_n
\end{bmatrix}, \\
\bfN &:=
\begin{bmatrix}
    \bfeta_1 &\cdots &\bfeta_n
\end{bmatrix}. 
\end{align*}
A natural choice of $\tilde\bfU$ is
\begin{equation}
    \label{eq:defUtilde2}
    \tilde \bfU:= \text{top $s$-left singular vectors of }\tilde \bfY.
\end{equation}
This choice of $\tilde\bfU$ is equivalent to taking the leading $s$ eigenvectors of the (normalized) covariance matrix
$$
\frac 1 n \tilde \bfY \tilde \bfY^* 
=\frac 1 n \sum_{\ell=1}^n \tilde \bfy_\ell \tilde \bfy_\ell^*
$$

To make this well-defined and a reasonable estimator for the correct subspace $\bfU$, we need several natural and technical assumptions which we collect from \cite[Assumption 3]{dirksen2025subspace}. 

\begin{assumption}
    \label{assump:specestimationproblem}
    For absolute constants $K>0$ and $\sigma>0$, assume the following hold. 
    \quad 
    \begin{enumerate}[(a)]
        \item 
        Assumption on noise. Let $\bfeta_1,\dots,\bfeta_n\overset{\text{i.i.d.}}{\sim} \bfeta$, where $\bfeta\in\C^m$ is a $K \sigma$-subgaussian distribution with uncorrelated entries that are mean zero and have variance $\sigma^2$. 
        \item 
        Assumption on amplitudes. Let $\bfa_1,\dots,\bfa_n\in \C^s$ be deterministic and assume the amplitude covariance matrix
        \begin{equation}
                \label{eq:ampcovariance}
                \bfSigma_{\bfa} :=\frac 1 n \sum_{\ell=1}^n \bfa_\ell \bfa_\ell^*
        \end{equation}
        has rank $s$. 
    \end{enumerate}
\end{assumption}

The rationale for these assumptions is well known. The assumption that the amplitude covariance matrix has full rank is required to ensure that $\E \tilde \bfY$ has rank $s$. The assumption on the noise implies that the population noise covariance matrix is $\sigma^2 \bfI$, which implies the leading $s$-eigenvectors of $\E \tilde \bfY$ is exactly $\bfU$, making $\tilde \bfU$ a consistent estimator as $n\to\infty$. \cref{assump:specestimationproblem} assumes the entries of $\bfeta$ are uncorrelated which is weaker than the standard assumptions used for direction-of-arrival estimation, such as those found in \cite{stoica1989music}. 

To control the subspace error, we appeal to a stochastic analogue of Wedin's sine-theta theorem \cite[Theorem 6]{dirksen2025subspace}, which is a general result that extends \cite[Theorem 3]{cai2018rate}. The former is a high probability version of the latter expectation bound and the former has relaxed statistical assumptions on the noise compared to the latter. The following lemma is shown in \cref{proof:stochasticWedin}.

\begin{lemma} \label{lem:stochasticWedin}
    Let $m\geq 100$ and $\{\theta_k\}_{k=1}^s\subset\T$ such that $\Delta(\{\theta_k\}_{k=1}^s)\geq 8\pi /m$. Suppose \cref{assump:specestimationproblem} holds for some $K>0$. If $s\leq \min\{m,n\}$ and $m\lambda_s(\bfSigma_{\bfa})\geq \sigma^2$, there are constants $c_K,\alpha_K$ that only depend on $K$ such that for all $\alpha \geq \alpha_K$, with probability at least $1-e^{-c_K \alpha m}$, the subspace $\tilde\bfU$ defined in \eqref{eq:defUtilde2} is well defined and
    $$
    \left\|\bfU\bfU^*- \tilde \bfU \tilde \bfU^* \right\|_2^2
    \lesssim \frac{ \alpha \sigma^2}{n \, \lambda_s(\bfSigma_{\bfa})}.  
    $$
\end{lemma}

\cref{lem:stochasticWedin} shows that the subspace error goes to zero as $n\to\infty$. Hence, by making $n$ large enough, we are in the setting where \cref{thm:step2} is applicable, which yields the third row of \cref{table:mainresults}. 

\begin{remark}[Technical comparison to prior work]
    \label{rem:techcomparison} 
    In \cite{rao1989performance}, Rao and Hari considered the multi-snapshot model, where the noise is a complex white Gaussian random vector with zero mean. It was found that \cite[(24), (29b)]{rao1989performance} the perturbation of the correct roots in the form of the mean squared error of $\Delta z_i$, which for us is $\left|e^{i\theta_k} - \tilde w_k\right|$, is 
    $$
    \E \left|e^{i\theta_k} - \tilde w_k\right| ^2 = \frac{(m-s)\sigma^2}{n (\bfphi'(e^{i\theta_k})^* \bfU_{\perp} \bfU_{\perp}^* \bfphi'(e^{i\theta_k}))} \sum_{j=1}^s \frac{\lambda_j + \sigma^2}{\lambda_j^2} \left|\bfphi(e^{i\theta_k})^* \bfu_j\right|, 
    $$ 
    where $\bfu_j$ is the $j$-th column of $\bfU$ and $\lambda_j = \lambda_j(\bfSigma_{\bfy})$ where $\bfSigma_{\bfy}=\bfPhi \bfSigma_{\bfa} \bfPhi^*$ is the covariance matrix of the noiseless data. For simplicity, assume that $|a_j|\asymp 1$ for each $j$.

    To simplify this expression, we make a few observations. First, a calculation shows that
    $$
    q''(\theta_k) = 2 \bfphi'(e^{i\theta_k})^* \bfU_{\perp} \bfU_{\perp}^* \bfphi'(e^{i\theta_k}), 
    $$
    where we recall that $q(t)=Q(e^{it})$ is the restriction of $Q$ to the circle. Recall that $q''(\theta_k)\gtrsim m^2$ under the condition $\Delta m \geq 8\pi$ due to \cref{lem:q''}. The reverse inequality $q''(\theta_k)\leq m^2$ holds by Bernstein inequality. Also under the same condition, $\sigma_j(\bfPhi)\asymp \sqrt m$ for each $1\leq j\leq s$ due to \cite{aubel2019vandermonde}, hence $\lambda_j \asymp m$ for each $j$. Then the previously mentioned result, when $m\gtrsim \sigma^2$ and $n\geq \max\{m,s\}$, simplifies to 
    \begin{align*}
        \E \left|e^{i\theta_k} - \tilde w_k\right| ^2 
        &= \frac{(m-s)\sigma^2}{n (\bfphi'(e^{i\theta_k})^* \bfU_{\perp} \bfU_{\perp}^* \bfphi'(e^{i\theta_k}))} \sum_{j=1}^s \frac{\lambda_j + \sigma^2}{\lambda_j^2} \left|\bfphi(e^{i\theta_k})^* \bfu_j\right|^2 \\ 
        &\asymp \frac{(m-s)\sigma^2}{n q''(\theta_k)} \frac{m + \sigma^2}{m^2} \sum_{j=1}^s \left|\bfphi(e^{i\theta_k})^* \bfu_j\right|^2 \\
        &= \frac{(m-s)\sigma^2}{n q''(\theta_k)} \frac{m + \sigma^2}{m^2} \left\|\bfU^* \bfphi(e^{i\theta_k})\right\|_2^2. 
    \end{align*}
    To simplify this expression further, note that $\bfphi(e^{i\theta_k})\in \range(\bfU)$ and so 
    $$
    \|\bfU^* \bfphi(e^{i\theta_k})\|_2
    =\|\bfU\bfU^*\bfphi(e^{i\theta_k})\|_2
    =\|\bfphi(e^{i\theta_k})\|_2 = 1.
    $$
    Thus, continuing where we had left off, 
    \begin{align*}
        \E \left|e^{i\theta_k} - \tilde w_k\right| ^2 
        &\asymp \frac{(m-s)\sigma^2}{n q''(\theta_k)} \frac{m + \sigma^2}{m^2} 
        \asymp \frac{\sigma^2}{n m^2}. 
    \end{align*}
    This is consistent with our results as in \cref{table:mainresults} under multi-snapshot frequency error. The main difference is that our result is a high probability statement, fully non-asymptotic in all parameters, and holds for general subgaussian noise. 
    

    Similarly \cite{krim1992operator} confirmed the results of \cite{rao1989performance} with similar noise assumptions in the multi-snapshot model: the root perturbation scales like $\sigma/(m\sqrt{n})$. However, it additionally found that the angular shift has order $\sigma/(m^{3/2}\sqrt{n})$. This means that in the setting described in this section, it may be possible to improve the frequency error even though we found sharp bounds for the total perturbation of the roots. This would necessarily require exploiting statistical properties of i.i.d. Gaussian random vectors and a more specialized analysis than the one given in this paper.
\end{remark}

\section{Numerical experiments}
\label{sec:experiments}
To find the zeros of $\tilde P$ closest to the unit circle, we start with at least $Cm (\log(m))^2$ uniformly spaced points on the unit circle and run Newton's algorithm for each one as initialization. This choice of starting points loosely follows the suggestions \cite{hubbard2001find}. If Newton's method for a given initialization converges and the value of $\tilde P$ is smaller than a chosen numerical tolerance parameter, that value is saved and is discarded otherwise. Among the saved outputs, the algorithm selects the $s$ ones (unique up to reflection) closest to the unit circle (in the sense of definition \eqref{eq:defdistance}). The arguments are the output of Root-MUSIC. 

Let us describe the Newton iteration step in more detail. 
To this end, we define the vector
\begin{equation}
    \bfpsi(z)
    = \frac1 {\sqrt m} \begin{bmatrix}
        z^{m-1} &z^{m-2} &\cdots &z & 1
    \end{bmatrix}^\top.
    \label{eq:defpsi}
\end{equation}
This is simply the vector $\bfphi(z)$ with its entries listed in reverse order. We have the following formulas for $\tilde P$ and $\tilde P'$. 

\begin{observation} \label{obs:P}
    For all $z\in \C$, we have 
    \begin{align}
        \tilde P(z)
        &= z^{m-1} - \bfpsi(z)^\top \tilde \bfU \tilde \bfU^* \bfphi(z), \label{eq:Ptilde1} \\
        \tilde P'(z)
        &= (m-1)z^{m-2} - \bfpsi'(z)^\top \tilde \bfU \tilde \bfU^* \bfphi(z) - \bfpsi(z)^\top \tilde \bfU \tilde \bfU^* \bfphi'(z).
        \label{eq:Ptilde2}    
    \end{align} 
\end{observation}

Note that these formulas only require computation of $\tilde \bfU$, which does not require the full SVD of a Hankel/Toeplitz/covariance matrix. Depending on which truncated SVD algorithm is used (deterministic versus random) and size of $s$ versus $m$, this saving is significant in practice. Moreover, $\tilde\bfU$ has fewer columns than $\tilde \bfU_\perp$ when $m\geq 2s$, so it advisable to use \eqref{eq:Ptilde1} and \eqref{eq:Ptilde2} over formulas that are in terms of $\tilde\bfU_\perp$. 

\begin{figure}[ht]
    \centering 
    \begin{subfigure}{0.44\textwidth}
        \includegraphics[width=1\linewidth]{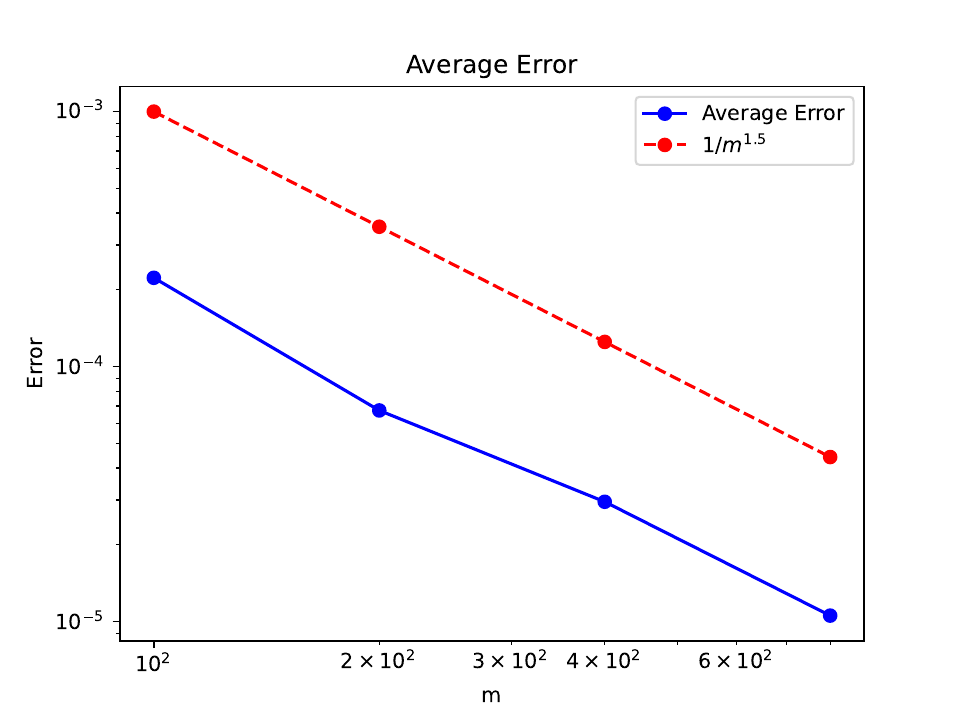}
    \caption{Experiment 1}
    \label{fig:average}  
    \end{subfigure}
    \begin{subfigure}{0.44\textwidth}
        \includegraphics[width=1\linewidth]{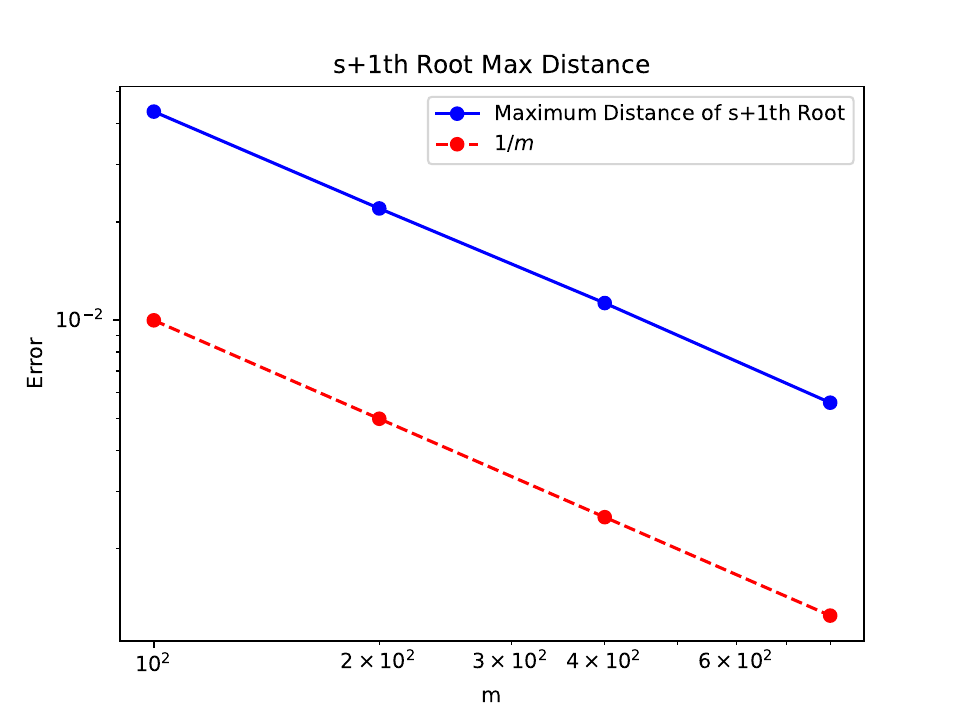}
    \caption{Experiment 2}
    \label{fig:s+1}
    \end{subfigure}
    \caption{Experiment 1: The average of the frequency error of the recovered frequencies $\max_k \left|\tilde w_k - e^{i\theta_k}\right|$. Experiment 2: The maximum distance of the $(s+1)$-th root of $\tilde P$ (where roots are in order of distance from the unit circle).}
    \label{fig:experiment}
\end{figure}

Two experiments were done to test the results of \cref{thm:step2} using the above initialization and selection criteria. In the experiments, the case where the noise $\bfeta$ is independent and identically distributed $\calN(0,\sigma^2)$ entries was tested, and the data was artificially produced with this noise where $\sigma=0.1$ and with a preselected set of three correct frequencies: $\{\theta_k\}_{k=1}^3 = \{2\pi/5, 3\pi/5, \pi\}$, with amplitudes $a_k = 1$ for $k=1,2,3$. The experiments were run with $m\in\{100,200,400,800\}$, and this was repeated over 20 trials per choice of $m$. 

\cref{fig:average} shows experiment 1: the average frequency error in log-log scale (filled line in blue), compared to the expected rate of $1/m^{1.5}$ (dotted line in red). The two lines being nearly parallel suggests that the rates are proportional. This confirms the rate in $m$ of the bound of the frequency error in the single-snapshot model with i.i.d. $\calN(0,\sigma^2)$ noise given in the second row of \cref{table:mainresults}. For experiment 2, \cref{fig:s+1} plots the max distance of the $(s+1)$-th root from the unit circle (mod reflection and ordered by distance from the unit circle) in log-log scale (filled line in blue), compared to the expected rate of $1/m$ (dotted line in red) predicted by \cref{thm:step2} part (b). The lines are fully parallel, and so this shows that the theorem is sharp.

We again emphasize that the main contribution of this paper is a theoretical analysis of the roots of $\tilde P$ generated by the Root-MUSIC algorithm, which has important implications to its performance. A full numerical study and professional implementation are beyond the scope of this paper. One direction of improvement  concerns the numerical stability of Root-MUSIC. Computing $z^m$ (hence $\bfphi(z)$ all other downstream terms) will result in numerical overflow for $|z|>1$ and large $m$. This means that a naive implementation of Root-MUSIC with standard Newton's method is likely to be numerically unstable when $m$ is sufficiently large. Since $|\tilde P(z)|\leq 1$ for $z\in\partial \D$, a calculation together with formula \eqref{eq:qdef2} shows that whenever $1\leq |z|\leq 1+\tau/m$, we have $\|\bfphi(1/z)\|_2\leq 1$ and
    $$
    |P(z)|
    \leq \|\bfphi(z)\|_2 \|\bfphi(1/z)\|_2 
    \leq \frac 1 {\sqrt m} \left( \sum_{j=0}^{m-1} |z|^{2j} \right)^{1/2}
    \lesssim \frac 1 {\sqrt \tau} \left(1+ \frac \tau m\right)^m
    \lesssim \frac 1 {\sqrt \tau} e^\tau.
    $$ 
    Consequently, Root-MUSIC is more stable than naively expected, in an annulus of thickness $\tau/m$ around the circle, which is the region of interest according to our theory. This suggests that numerical instability can be tempered by employing a Newton's method with appropriate exit termination when an iterate leaves such an annulus.

\section{Proofs}
\label{sec:proofs} 

\subsection{Lemmas for Step 1a}

As explained in \cref{sec:outline}, we will lower bound $|P|$ by approximating it by a quadratic function, which involves dealing with $|P''(w_k)|$. We recall the following result in \cite[Lemma 9.4]{fannjiang2025optimality}.

\begin{lemma} \label{lem:q''}
    Suppose $m\geq 100$ and $\Delta(\{\theta_k\}_{k=1}^s)\geq 8\pi/m$. There is an absolute constant $c_0\in (0,1/6)$ such that for all $k$, 
    $$
    q''(\theta_k) \geq c_0 m^2. 
    $$
\end{lemma}

Although the $c_0$ that appears in this lemma is explicit and can be computed, we will not need its exact value. The important part of this inequality is that $c_0$ does not depend on either $m$ or the cardinality of $\{\theta_k\}_{k=1}^s$. 

The next lemma truncates the power series expansion of $P$ at $w_k$ and shows that third and higher order terms are negligible whenever $|z-w_k|\leq c_0/(5m)$. But first, it requires the following fact regarding polynomials.

\begin{theorem}\label{thm:bernstein}
    If a polynomial $P(z)$ has degree $n$, then the following is Bernstein's inequality. 
    \[\max_{|z|=1}|P'(z)| \leq n\cdot \max_{|z|=1}|P(z)|\]
\end{theorem}

With these two results, we prove the following lemma. 

\begin{lemma}\label{lem:Ppowerdisk}
    Suppose $m\geq 100$ and $\Delta(\{\theta_k\}_{k=1}^s)\geq 8\pi/m$. If $z\in D(w_k,c_0/(5m))$, then 
    $$
    |P(z)|
    \geq \frac 12 \left( 1 - \frac{8}{15} e^{2c_0/5}\right) c_0 m^2 |z-w_k|^2. 
    $$
\end{lemma}

\begin{proof}
    We plan to calculate the power series expansion of $P$ at its double root $w_k=e^{i\theta_k}$. 
    We have
	\begin{equation}
		\label{eq:Pseries}
		P(z)=\frac 12 P''(w_k)(z-w_k)^2 + \sum_{j=3}^{2m-2} \frac{P^{(j)}(w_k)}{j!} (z-w_k)^j. 
	\end{equation}
	By Bernstein's inequality for complex polynomials, we see that 
	$$
	\max_{z\in \partial\D} |P^{(j)}(z)|
	\leq (2m-2)^j \max_{z\in \partial\D} |P(z)|
	= (2m-2)^j \|q\|_{L^\infty(\T)}
	\leq (2m-2)^j. 
	$$ 
	Using this, we see that
	\begin{align*}
		\left| \sum_{j=3}^{2m-2} \frac{P^{(j)}(w_k)}{j!}(z-w_k)^j \right|
		&\leq 8(m-1)^3|z-w_k|^3 \, \sum_{j=3}^{2m-2} \frac{2^{j-3}(m-1)^{j-3} |z-w_k|^{j-3}}{j!} \\
		&\leq 8(m-1)^3|z-w_k|^3 \, \sum_{j=0}^\infty \frac{2^j(m-1)^j |z-w_k|^j}{(j+3)!} \\
        &\leq \frac{4}{3} (m-1)^3|z-w_k|^3 \, \sum_{j=0}^\infty \frac{2^j(m-1)^j |z-w_k|^j}{j!} \\
		&= \frac{4}{3} (m-1)^3 \, |z-w_k|^3 \, e^{2 (m-1)|z-w_k|}. 
	\end{align*}

    Additionally, \cref{obs:3} implies that $Q(w_k)=Q'(w_k)=0$ since each $w_k$ is a root with multiplicity two of $Q$. Then a calculation shows that 
    $$
    P''(w_k)=w_k^{m-1}Q''(w_k) \andspace q''(\theta_k)=-w_k^2 \, Q''(w_k).
    $$
    By \cref{lem:q''}, there is a $c_0\in (0,1/6)$ such that $q''(\theta_k)\geq c_0m^2$. Thus,
	$$
	|P''(w_k)|
	=|Q''(w_k)|
	=|q''(\theta_k)|
	\geq c_0 m^2.
	$$
	Thus, inserting the above inequality into \eqref{eq:Pseries}, we have
	\begin{align*}
		|P(z)|
		&\geq \frac 12 |P''(w_k)||z-w_k|^2 - \left| \sum_{j=3}^{2m-2} \frac{P^{(j)}(w_k)}{j!}(z-w_k)^j \right| \\
		&\geq \frac 12 c_0 m^2 |z-w_k|^2- \frac 4 3(m-1)^3|z-w_k|^3 e^{2(m-1)|z-w_k|} \\
		&\geq \frac 12 m^2 |z-w_k|^2 \left( c_0 - \frac 83 m|z-w_k| e^{2 m|z-w_k|}\right).
	\end{align*}
    By assumption, $|z-w_k|\leq c_0/(5m)$, and so 
    $$
    c_0-\frac 8 3m|z-w_k| e^{2 m|z-w_k|}
    \geq c_0-\frac{8}{15} e^{2c_0/5} c_0.
    $$
    Combining the above completes the proof. 
\end{proof}

\subsection{Lemmas for Step 1b}

We start with some bounds for $q$, before extending this result to the complex plane. Using the formula for $q$ in \eqref{eq:qdef2}, we obtain the global and unconditional perturbation bound,
\begin{equation}
	\label{eq:qerror1}
	\max_{t\in \T} |q(t)-\tilde q(t)|
    \leq \|\bfU_\perp \bfU_\perp^* - \tilde \bfU_\perp \tilde \bfU_\perp^*\|_2 
	=\|\bfU \bfU^* - \tilde \bfU \tilde \bfU^*\|_2 
	=\rho.  
\end{equation}
As mentioned in \cref{sec:outline}, we need a stronger inequality to invoke Rouche's theorem. We recall the following (also unconditional) inequality \cite[Lemma 9.8]{fannjiang2025optimality}, 
\begin{equation}
	\label{eq:qerror2}
	\max_k |q(\theta_k)-\tilde q(\theta_k)|\leq \rho^2. 
\end{equation}
To control $|P(z)-\tilde P(z)|$ for $z$ close enough to a $w_k$, we develop the following perturbation inequalities that will be useful throughout the proofs.

\begin{lemma} \label{lem:Pradialdisk}
	For any $k\in \{1,\dots,s\}$, $m>s$, and $z\in \overline \D$, we have 
	\begin{align*}
		|P(z)-\tilde P(z)|
		&\leq 2m\rho \, |z-w_k| + \rho^2, \\ 
        |P(1/\overline z)-\tilde P(1/\overline z)|
		&= |z|^{-2m+2} |P(z)-\tilde P(z)|. 
	\end{align*}    
\end{lemma}

\begin{proof}
    It will be helpful to define $H=P-\tilde P$, which is also a polynomial of degree at most $2m-2$. Fix any $z\in \overline \D$. Let $\gamma$ be the straight line segment contained in $\overline \D$ that connects $w_k$ and $z$. Then 
    $$
    H(z)=H(w_k)+\int_\gamma H'(w)\, dw.  
    $$
    By the max modulus principle and Bernstein's inequality, we see that
    \begin{align*}
        |H(z)|
        &\leq |H(w_k)|+ |z-w_k| \, \max_{w\in \D} |H'(w)| \\
        &\leq |H(w_k)|+ |z-w_k| \, (2m-2) \max_{w\in \partial\D}  |H(w)|.
    \end{align*}
    Observe that the definition of $H$ and inequalities \eqref{eq:qerror1} and \eqref{eq:qerror2} imply
    \begin{align*}
        |H(w_k)|
        &=|q(\theta_k)-\tilde q(\theta_k)|
        \leq \rho^2, \\
        \max_{w\in \partial\D} |H(w)|
        &=\max_{t\in \T} |q(t)-\tilde q(t)|
        \leq \rho.
    \end{align*}
	Combining the previous two displayed inequalities proves the first inequality of this lemma. 

    For the second inequality, recall Observations \ref{obs:3} and \ref{obs:4}. Then
	\begin{align*}
		P(1/\overline z)-\tilde P(1/\overline z)
		&= \left(1/\overline z \right)^{m-1} \left( Q(1/\overline z)-\tilde Q(1/\overline z) \right)\\
		&= \left(1/\overline z \right)^{m-1}\left( \overline{Q(z)}- \overline{\tilde Q(z)} \right)\\
		&= \left(1/\overline z \right)^{2m-2} \left( \overline{P(z)}-\overline{\tilde P(z)} \right). 
	\end{align*}
	Taking the absolute value completes the proof. 
\end{proof}

We need the following geometric observation.
\begin{lemma}
    \label{lem:reflection}
    Suppose $z\in \C$ with $|z|>1$. Then for any $w\in \partial\D$, it holds that $|1/\overline z-w|\leq |z-w|$. 
\end{lemma}
\begin{proof}
    For convenience, let $R:=|z|>1$ and $\alpha$ denote the angle between $z$ and $w$, which is also the angle between $1/\overline z$ and $w$. Applying the law of cosines for the triangle with vertices $\{0,w,1/\overline z\}$,
    \begin{align*}
        \left| \frac 1{\overline z} - w \right|^2 
        &=1+\frac 1 {R^2} - \frac{2}{R} \cos (\alpha).
    \end{align*}
    Applying the law of cosines for the triangle with vertices $\{0,w,z\}$, we see that
    $$
    |z-w|^2 
    = 1+R^2 -2R\cos(\alpha).
    $$
    Then the conclusion of this lemma is equivalent to 
    $$
    R^2-2R\cos(\alpha) - \frac 1 {R^2} + \frac 2 R \cos(\alpha) \geq 0.
    $$
    Multiplying both sides by $R^2$ and then dividing by the positive quantity $R^2-1$, we see that the above is equivalent to 
    $$
    R^2+1-2R\cos(\alpha) \geq 0.
    $$
    This is true for all $\alpha$ since
    $$
    R^2+1-2R\cos(\alpha)
    \geq R^2+1-2R
    \geq 0.
    $$
    This completes the proof. 
\end{proof}

By using these perturbation bounds, we can control the perturbation of $P$ within a neighborhood of $w_k$. 

\begin{lemma} \label{lem:Perrorcircle}
    For any $k\in \{1,\dots,s\}$, $m>s$, and $z\in D(w_k,r)$ for some $r>0$, we have
    $$
    |P(z)-\tilde P(z)| 
    \leq (1+r)^{2m-2} \left( 2m \rho r + \rho^2 \right). 
    $$
\end{lemma}

\begin{proof}
    We first consider the case where $z\in D(w_k,r)$ and $|z|\leq 1$. Applying the first inequality in \cref{lem:Pradialdisk}, we see that
    $$
    |P(z)-\tilde P(z)| 
    \leq 2m \rho r + \rho^2. 
    $$
    
    For the remaining case where $z\in D(w_k,r)$ and $|z|\geq 1$, it follows from \cref{lem:reflection} that $1/\overline z \in D(w_k,r)$ as well.
    By both parts of \cref{lem:Pradialdisk},
    \begin{align*}
        |P(z)-\tilde P(z)| 
        &\leq |z|^{2m-2} |P(1/\overline z) - \tilde P(1/\overline z)| \\
        &\leq (1+r)^{2m-2} \left( 2m \rho |1/\overline z - w_k| + \rho^2 \right) \\
        &\leq (1+r)^{2m-2} \left( 2m \rho r + \rho^2 \right).
    \end{align*}
    This completes the proof. 
\end{proof}

By employing these results, we first show that for each $\theta_k$, $\tilde P$ has exactly two roots with reflection symmetry that are in $D(w_k,\alpha/m)$ for an absolute constant $\alpha>0$.
\begin{lemma}\label{lem:coarseroots}
    There are absolute constants $\alpha,\epsilon_2>0$ such that the following hold. Let $m\geq 100$ and $\{\theta_k\}_{k=1}^s\subset\T$ such that $\Delta(\{\theta_k\}_{k=1}^s)\geq 8\pi/m$ and $\tilde\bfU$ such that 
    $$
    \rho:=\|\bfU\bfU^*-\tilde\bfU\tilde\bfU^*\|_2\leq \epsilon_2.
    $$
    The Root-MUSIC polynomial $\tilde P$ has exactly two roots in the disk $D(w_k,\alpha/m)$ for each $k\in \{1,\dots,s\}$.
\end{lemma}

\begin{proof}
    Fix $w_k=e^{i\theta_k}$ and consider the circle $\partial D(w_k, \alpha/m)$ where $\alpha>0$ will be chosen later. For now, assume that $\alpha\leq c_0/5<1/30$. By \cref{lem:Ppowerdisk}, for all $z\in \partial D(w_k, \alpha/m)$, we have 
	\begin{equation}
		\label{eq:Plower2}
		|P(z)| 
		\geq \frac 12 \left( 1 - \frac{8}{15} e^{2c_0/5}\right) c_0 \alpha^2. 
	\end{equation}	
    On the other hand, by \cref{lem:Perrorcircle}, the inequality $1+t\leq e^t$, and the assumption $\rho\leq \epsilon_2$, for all $z\in \partial D(w_k, \alpha/m)$, 
	\begin{equation}
		\label{eq:Pupper2}
		|P(z)-\tilde P(z)|
		\leq \left(1 + \frac{\alpha} m \right)^{2m-2} (2 \alpha \epsilon_2 +\epsilon_2^2) 
        \leq e^{2\alpha} (2\alpha \epsilon_2+\epsilon_2^2)
        \leq e^{2c_0/5} (2\alpha \epsilon_2 + \epsilon_2^2). 
	\end{equation}
    Now we are in position to pick $\alpha$ such that the right side of \eqref{eq:Plower2} is strictly larger than the right side of \eqref{eq:Pupper2}. This clearly possible for $\alpha$ large enough (since the lower bound is quadratic while the upper bound is linear) and explicitly,
    \begin{equation}
        \label{eq:alphalower}
        \alpha > \frac{2e^{2c_0/5}\epsilon_2 + \epsilon_2 \sqrt{4e^{4c_0/5}+2c_0(1 - 8e^{2c_0/5}/15)e^{2c_0/5}}}{c_0(1 - 8e^{2c_0/5}/15)}.
    \end{equation}
    Notice that \eqref{eq:alphalower} requires $\alpha \geq C\epsilon_2$ for some absolute $C>0$. Thus, make $\epsilon_2$ even smaller if necessary so that there is a $\alpha<c_0/5$ that also satisfies \eqref{eq:alphalower}. Then 
    $$
    |P(z)-\tilde P(z)|
    < |P(z)| \forallspace z\in \partial D(w_k, \alpha/m).
    $$
    The assumptions of Rouche's theorem hold, so we conclude that the two polynomials $P$ and $\tilde P$ have the same number of zeros inside the region enclosed by $\partial D(w_k, \alpha/m)$. 
\end{proof}

The reader may notice that \cref{lem:coarseroots} appears worse than what is stated in \cref{thm:step1}. The usefulness of the lemma will be apparent in the proof of the theorem and again in Step 2. It is also important to note that \cref{lem:coarseroots} does not assert that the two roots of $\tilde P$ in $D(w_k,\alpha/m)$ is a pair with reflection symmetry.  

\subsection{Proof of \cref{thm:step1}}
\label{proof:step1}

\begin{proof}
	Let $\alpha,\epsilon_2>0$ be the quantities in \cref{lem:coarseroots}. \commentout{We concentrate on the first assertion of this theorem. }We will pick $R$ and $\epsilon_1$ appropriately later. For now, we make $R\geq 1$, $\epsilon_1\leq \epsilon_2$, and $R\epsilon_1\leq c_0/5$, where $c_0\in (0,1/6)$ is given in \cref{lem:q''}. By \cref{lem:Ppowerdisk}, for all $z\in \partial D(w_k,R\rho/m)$, we have 
	\begin{equation}
		\label{eq:Plower}
		|P(z)| 
		\geq \frac 12 \left( 1 - \frac{8}{15} e^{2c_0/5}\right) c_0 R^2 \rho^2. 
	\end{equation}	
    Note $c_0<1/6$, so the right hand side is a non-vacuous bound. On the other hand, by \cref{lem:Perrorcircle} and the inequality $1+t\leq e^t$, for all $z\in \partial D(w_k,R\rho/m)$, 
	\begin{equation}
		\label{eq:Pupper}
		|P(z)-\tilde P(z)|
		\leq \left(1 + \frac{R\rho} m \right)^{2m-2} (2 R\rho^2 +\rho^2) 
        \leq e^{2R\rho} (2R+1)\rho^2
        \leq e^{2c_0/5} (2R+1)\rho^2. 
	\end{equation}
	Now we are in position to choose $R$ correctly in order to invoke Rouche's theorem. Notice that the right hand side of \eqref{eq:Plower} is quadratic in $R$ while the right hand side of \eqref{eq:Pupper} is linear in $R$, and they both have the same scaling in $\rho$. Explicitly, pick any $R$ such that 
    \begin{equation}
        \label{eq:Rcondition}
        \frac{R^2}{2R+1}> \frac 2 {c_0} \left(\frac{e^{2c_0/5}}{1-8e^{2c_0/5}/15}\right), 
    \end{equation}
    and let $\epsilon_1$ such that $R\epsilon_1 \leq c_0/5$ still holds. For these choices of parameters, using \eqref{eq:Plower} and \eqref{eq:Pupper}, we get
    $$
    |P(z)-\tilde P(z)|
    < |P(z)| \forallspace z\in \partial D(w_k,R\rho/m).
    $$
	The assumptions of Rouche's theorem hold, so we conclude that the two polynomials $P$ and $\tilde P$ have the same number of zeros inside the region enclosed by $\partial D(w_k,R\rho/m)$. Since $R\rho\leq R\epsilon_1\leq c_0/5 < 1/30$ while $\Delta(\{\theta_k\}_{k=1}^s)\geq 8\pi/m$, we see that $P$ has a single double root inside $\partial D(w_k,R\rho/m)$, namely $w_k$. Thus, $\tilde P$ has exactly two roots inside $\partial D(w_k,R\rho/m)$ as well. 
    
    By \eqref{eq:Pdef} the roots of $\tilde P$ possess the reflection symmetry given in \cref{obs:4}, so it remains to show that the two roots of $\tilde P$ in the disk $D(w_k, R\rho /m)$ can be guaranteed to be of the form $z_k$ and $1/\overline{z_k}$ whenever $\epsilon_1$ is sufficiently small. To see why this is indeed the case, if $z_k\in D(w_k,R\rho/m)$ is a root of $\tilde P$ and $z_k\not\in \D$, then we are done since $1/\overline {z_k}\in D(w_k,R\rho/m)$ as well due to \cref{lem:reflection}. Consider the case that $z_k\in D(w_k,R\rho/m)$ and $z_k\in \D$. Recall \cref{lem:coarseroots} which guarantees that $\tilde P$ has exactly two roots in $D(w_k,\alpha/m)$ for an absolute $\alpha>0$. Then the proof is complete once we show that $1/\overline{z_k}\in D(w_k,\alpha/m)$. We claim that this is indeed the case if $\epsilon_1$ is sufficiently small.
    Noting that $\left|1/\overline{z_k} - e^{i\arg(z_k)} \right| = \left|\left|1/\overline{z_k}\right| e^{i\arg(z_k)}  - e^{i\arg(z_k)} \right| = \left|1/\overline{z_k}\right| -1$, we obtain: 

    \begin{align*}
    \left|\frac1{\overline{z_k}} - w_k\right|
    &\leq \left|\frac1{\overline{z_k}} - e^{i\arg(z_k)} \right| + \left| e^{i\arg(z_k)} - e^{i\theta_k} \right| \\
    &\leq \left|\frac1{\overline{z_k}} \right| -1 + \left|\arg(z_k)-\theta_k\right| \\
    &\leq \left(1-\frac{R\rho}m\right)^{-1}-1 + \frac{R\rho}m \\
    &= \frac {R \rho} {m-R\rho} + \frac {R\rho}m.
    \end{align*}
    Since $m\geq 100$, whenever $R\epsilon_1\leq 50$, we have $m-R\rho \geq m/2$. Now we need to make $\epsilon_1$ even smaller if necessary so that $4R\epsilon_1\leq \alpha$. Then inserting this into the previous inequality shows that 
    $$
    \left|\frac1{\overline{z_k}} - w_k\right|
    \leq \frac{3R\rho}m 
    \leq \frac{3R\epsilon_1}m
    < \frac \alpha m. 
    $$
    This implies both $z_k \in \D$ and its reflection are the only two roots contained in $D(w_k,\alpha/m)$ and consequently in $D(w_k,R\rho/m)$. 

    To summarize the parameter choice argument, note that $\alpha,\epsilon_2$ are fixed absolute constants. We first pick $R\geq 1$ to be a sufficiently large absolute constant so that inequality \eqref{eq:Rcondition} holds. Then we pick $\epsilon_1$ sufficiently small so that $\epsilon_1\leq \epsilon_2$ and  $R\epsilon_1\leq \min\{c_0/5, \alpha/4, 100\}$.  
\end{proof}

\subsection{Lemmas for Step 2a}

In the expository parts, we assume that the assumptions and conclusions of \cref{thm:step1} hold. By the first part of the theorem, $\tilde P$ has two roots in $D(w_k,R\rho/m)\subset D(w_k, R\epsilon_1/m)$. 

\begin{figure}[ht]
        \centering
        \begin{subfigure}{0.44\textwidth}
    \begin{tikzpicture}
        \fill[red] (0,0) ++(45:5) circle [radius=2pt] coordinate (wk);
        \node[red] (w) at (4.5,2) {$w_k$};
        \draw[->, red, thick] (0,0)++(45:5) .. controls +(right:0cm) and +(left:1cm) .. (w);
        
        \draw[red!50!blue, ultra thick] (wk) circle (2.5);
        \draw[red!70!blue, ultra thick] (wk) circle (0.6);

        \draw[blue, ultra thick] (0,0) ++(35:5.8) coordinate (a) arc (35:55:5.8) coordinate (c);
        \draw[blue, ultra thick] (0,0) ++(35:4.2) coordinate (b) arc (35:55:4.2) coordinate (d);
        \draw[blue, ultra thick] (a) -- (b);
        \draw[blue, ultra thick] (d) -- (c);

        \draw[blue, very thick, densely dotted] (0,0)++(45:4.2) -- (wk);
        \node[blue] (e) at (2.5,2.5) {$\frac {\tau_1}m$};
        \draw[->, blue, thick] (0,0)++(45:4.6) .. controls +(right:0cm) and +(right:1cm) .. (e);
        \draw[blue, very thick, densely dotted] (0,0)++(35:5) -- (wk); 
        \node[blue] (f) at (4.7,4.5) {$2\sin(\frac {\alpha_1}{2m})$};
        \draw[->, blue, thick] (0,0)++(40:5) .. controls +(left:0cm) and +(down:1cm) .. (f);
        \draw[red!50!blue, dashed, very thick] (wk) ++ (0,2.5) -- node[right, near start] {$\frac {\alpha}m$} (wk);
        \draw[red!70!blue, dashed, very thick] (wk) ++ (140:0.6) -- (wk);
        \node[red!70!blue] (g) at (2,4.5) {$\frac {R\epsilon_3}m$};
        \draw[->, red!70!blue, thick] (wk) ++ (140:0.3) .. controls +(right:0cm) and +(down:1cm) .. (g);
    \end{tikzpicture}
    \caption{$D(w_k,R\epsilon_3/m) 
    \subset W(w_k,1-\tau_1/m,\alpha_1/m)
    \subset D(w_k,\alpha/m)$, using (b) where $\theta = \frac{\alpha_1}m$.} 
        \label{fig:figure2}
        \end{subfigure}
    \qquad
        \begin{subfigure}{0.44\textwidth}
            \begin{tikzpicture}
        \fill[black!70] (0,0) circle [radius=2pt] node[right] {$(0,0)$};
        \fill[black!70] (0,0) ++(60:5) circle [radius=2pt] node[above right = -2pt] {$w_k$} coordinate (wk);
        \fill[black!70] (0,0) ++(120:5) circle [radius=2pt] node[above left = -3.5pt] {$w_k e^{i\theta}$} coordinate (w);
        \draw[black, thick] (wk) -- node[below,sloped] {$r=1$} (0,0);
        \draw[black, thick] (w) -- (0,0);
        \draw[black, dotted, thick] (wk) -- coordinate[midway](m) (w);
        \draw[black, dotted, thick] ++(90:5) -- (0,0);
        \draw[black,thick] (m)+(-0.5,0) -- ++(-0.5,-0.5) -- ++(0.5,0);
        \draw[red, thick] ++(60:1.5) arc (60:120:1.5) node[midway, above] {$\theta$ rad};
        \draw[blue, very thick] ++(60:5) arc (60:120:5) node[midway, above] {arc length $\theta$};
        \draw[red!50!blue, very thick] (w) -- node[midway,below] {$\sin(\frac{\theta}{2})$} (m);
    \end{tikzpicture}
    \caption{Given a difference in angle $\theta$ radians, the distance between $w_k$ and another point on the unit circle $w_ke^{i\theta}$ is $2\sin(\frac{\theta}2)$.}
        \label{fig:figure1} 
        \end{subfigure}
    \caption{Figures for proof of \cref{obs:5}.}
    \end{figure}
    
\begin{lemma}
    \label{obs:5}
    Suppose the assumptions and conclusions of \cref{thm:step1} hold. There exist $\tau_1,\alpha_1,m_0>0$ depending only on $\alpha$ and $\epsilon_3>0$ depending only on $\alpha, R$, such that for $m \geq m_0$,
    $$
    D(w_k,R\epsilon_3/m) 
    \subset W(w_k,1-\tau_1/m,\alpha_1/m)
    \subset D(w_k,\alpha/m). 
    $$
\end{lemma}

\begin{proof}
    Let $\tau_1, \alpha_1 >0$ be any pair such that $\alpha > 2\tau_1 + \alpha_1$, and let $m_0 > 2\tau_1$. Note that by this construction, since $\tau_1$ only depends on $\alpha$, so does $m_0$. 
    We first inscribe $W_k := W(w_k,1-\tau_1/m,\alpha_1/m)$ into $D(w_k,\alpha/m)$. Note that $\left( 1-\frac{\tau_1}m \right)^{-1} -1= \frac{\tau_1}{m-\tau_1} > \frac{\tau_1}m$ for any $0<\tau_1<m$. Then $m>2\tau_1$ and $m-\tau_1>\frac m2$ for all $m\geq m_0$. By a geometric argument as seen in \cref{fig:figure2}, the criteria for $z$ on the wedge implies that for all $z\in \partial W_k$,
    \begin{align*}
        |z-w_k| 
        &\leq |z - e^{i\arg(z)}| + |e^{i\arg(z)} - w_k| 
        \leq \frac {\tau_1}{m-\tau_1}  + \frac {\alpha_1}m \\ 
        &< \frac{\tau_1}{m/2} + \frac {\alpha_1}m 
        = \frac{2\tau_1}m + \frac {\alpha_1}m 
        < \frac \alpha m. 
    \end{align*}
    This shows that the disk of radius $\alpha/m$ centered at $w_k$ contains the wedge $W_k$. 
    
    Next, we inscribe $D(w_k,R\epsilon_3/m)$ into $W_k$. It is safe to assume that $\epsilon_3$ is sufficiently small such that $R\epsilon_3 < \min \left\{ \tau_1, \frac{2\alpha_1}\pi \right\}$, thus depending only on $\alpha$ and $R$. By another geometric argument as in \cref{fig:figure1} and using that $\sin(t)\geq 2t/\pi$ for all $t\in [0,\pi/2]$, we obtain for all $z\in \partial W_k$,
    \begin{align*}
        |z-w_k| &\geq \min\Big\{2\sin\left(\frac{\alpha_1}{2m}\right), \, \frac {\tau_1} m, \, \frac 1{1-\tau_1/m} -1 \Big\} \\
        &\geq \min\Big\{\frac{2\alpha_1}{\pi m}, \frac {\tau_1}m \Big\} 
        > \frac {R\epsilon_3}m. 
    \end{align*}
    Therefore the wedge $W_k$ contains $D(w_k,R\epsilon_3/m)$. 
\end{proof}

Combining this observation with \cref{thm:step1}, we deduce $\tilde P$ has two roots with reflection symmetry in the wedge $W_k = W(w_k,1-\tau_1/m,\alpha_1/m)$. It remains to show that $\tilde P$ has no other roots in the rest of an annulus that does not intersect $\bigcup_k W_k$. 

Reference \cite[Lemma 9.6]{fannjiang2025optimality} provided a lower bound for $q(t)$ under the condition that $d(t,\theta_k)\geq \pi/m$ for all $\theta_k\in\{\theta_k\}_{k=1}^s$ and $\Delta(\{\theta_k\}_{k=1}^s)\geq 2\pi \beta/m$ for fixed $\beta>1$. This was then used to deduce that $q(t)\geq c$ for some absolute constant $c>0$. Unfortunately we cannot directly quote this result because we need to lower bound $q(t)$ whenever $d(t,\theta_k)\geq \alpha_1/m$ where $\alpha_1$ is potentially smaller than $\pi$. This is just a technicality and following the same arguments, we can get the same conclusion with a possibly different constant.  

\begin{lemma}
    \label{lem:far1}
    Fix any integer $m_0\geq 1$ and any $\alpha_2>0$. There exist $c_1>0$ and $\beta\geq 1$ depending only on $m_0,\alpha_2$ such that for any $\{\theta_k\}_{k=1}^s\subset\T$ such that $m\geq m_0$ and $\Delta(\{\theta_k\}_{k=1}^s)\geq 2\pi\beta/m$, the following holds. For all $t\in \T$ with $\min_k d(t,\theta_k)\geq \alpha_2/m$, we have $q(t)\geq c_1$.
\end{lemma}

\begin{proof}
    Fix any $t$ satisfying the assumption and let $\theta_k$ be an element in $\{\theta_k\}_{k=1}^s$ that is closest to $t$. We can assume without loss of generality that $d(t,\theta_k)<\pi/m$, otherwise we use \cite[Lemma 9.6]{fannjiang2025optimality} and we are done. 
    
    We will use \cite[Lemma 9.3]{fannjiang2025optimality}, which provides a local approximation of $q(t)$ in terms of a normalized Fej\'er kernel,
    $$
    f_m(t) := \frac 1 {m^2} \left( \frac{\sin(mt/2)}{\sin(t/2)}\right)^2.
    $$
    The assumptions of the referenced lemma requires $A(\beta) E_0(m_0,\beta,\beta) < 1/4$, where $A(\beta)$ and $E_0(m_0,\beta,\beta)$ are explicit quantities that do not depend on $m$ and are defined in \cite[Section 9]{fannjiang2025optimality}. We do not need their explicit formulas, but we only need to note that $A(\beta)\to 1$ and $E_0(m_0,\beta,\beta) \to 0$ as $\beta\to\infty$. Hence, for sufficiently large absolute $\beta$, we have $A(\beta) E_0(m_0,\beta,\beta) < 1/4$. Then
    $$
    q(t)=1-f_m(t-\theta_k)-\bfphi(t)^* \bfW_k \bfW_k^* \bfphi(t),
    $$
    where $\bfW_k\in \C^{m\times (s-1)}$ was defined in the referenced lemma. Again, we do not need to know what $\bfW_k$ is, except that repeating the same argument in the proof of \cite[Lemma 9.6]{fannjiang2025optimality}, when $\beta\to\infty$, 
    $$
    \bfphi(t)^* \bfW_k \bfW_k^* \bfphi(t)
    = \|\bfW_k^* \bfphi(t)\|_2^2
    \to 0, 
    $$
    uniformly over $m\geq m_0$ and $t\in \T$ such that $d(t,\theta_k)\geq \alpha_2/m$ for all $k$. Next, we use local approximation of the Fej\'er kernel such as \cite[Lemma 9.10]{fannjiang2025optimality} to see that since $\alpha_2/m\leq d(t,\theta_k)\leq \pi/m$, there is a $c_1>0$ that depends only on $\alpha_2$ and not on $m$ such that 
    $$
    1-f_m(t-\theta_k)\geq 2c_1. 
    $$
    Hence, we pick $\beta$ big enough depending only on $\alpha_2$ and $c_1$ so that $q(t)\geq c_1$, which completes the proof. 
\end{proof}

Having controlled $q(t)$ for $t$ that are sufficiently far away from $\{\theta_k\}_{k=1}^s$, we need to control $P$ on the annulus in \eqref{eq:naturalannulus}.

\begin{lemma}
    \label{lem:far2}
    Fix any integer $m_0\geq 1$ and let $\alpha_2>0$. There exist $\epsilon_4, \tau_2>0$ and $\beta\geq 1$ depending only on $m_0,\alpha_2$ such that the following holds. For any $\rho \leq \epsilon_4$, $m\geq m_0$, and $\{\theta_k\}_{k=1}^s\subset\T$ such that $\Delta(\{\theta_k\}_{k=1}^s)\geq 2\pi\beta/m$, $|\tilde P(z)|>0$ for all $z\in \C$ such that $1-\tau_2/m \leq |z|\leq (1-\tau_2/m)^{-1}$ and $\min_k d(\arg(z),\theta_k)\geq \alpha_2/m$. 
\end{lemma}

\begin{proof}
    For $t=arg(z)$, when $\min_k d(\arg(z),\theta_k)\geq \alpha_2/m$, the assumptions of \cref{lem:far1} hold so there is a $c_1>0$ depending only on $m_0$ and $\alpha_2$ such that $|P(e^{it})|=|Q(e^{it})|=q(t)\geq c_1$. 
    Note that for all such $t$, using \eqref{eq:qerror1}, 
    $$
    \tilde q(t) \geq q(t)-|\tilde q(t)-q(t)| 
    \geq c_1 -\rho.
    $$
    Additionally, let $\gamma$ be the line segment from $w$ to $rw$ where $r<1$ and $w \in \mathbb{C}$ such that $|w|=1$. Then by the maximum modulus principle, 
    \[
    |\tilde P(rw) - \tilde P(w)|
    =\left|\int_{\gamma} \tilde P'(z)\, dz \right|
    \leq (1-r) \max_{z\in \overline \D} |\tilde P'(z)| \leq (1-r) \max_{|z|=1} |\tilde P'(z)|.
    \]
    Thus, by Bernstein's inequality, 
    \[
    \frac{|\tilde P(rw) - \tilde P(w)|}{1-r} \leq \max_{|z|=1} |\tilde P'(z)| \leq (2m-2) \max_{|z|=1} |\tilde P(z)| \leq 2m-2.
    \]
    Therefore, if $z=rw$ with $r<1$ and $w\in \partial\D$, 
    \begin{align*}
        |\tilde P(rw)| 
        &\geq |\tilde P(w)| - |\tilde P(rw) - \tilde P(w)| \\
        &= |\tilde q(\arg(w))| - |\tilde P(rw) - \tilde P(w)| \\
        &\geq (c_1-\rho) - (1-r) (2m-2) \\
        &\geq (c_1-\rho) - \frac {\tau_2} m (2m-2) \\
        &\geq c_1 - \epsilon_4 - 2\tau_2 
    \end{align*}
    Note $c_1$ only depends on $m_0$ and $\alpha_2$. Then $\epsilon_4$ and $\tau_2$ can be chosen appropriately depending only on $c_1$, hence depending only on $m_0$ and $\alpha_2$, so that this is strictly greater than zero. 
    
    For $z=r^{-1}w$, since $\tilde Q(rw) = \overline{\tilde Q(r^{-1}w)}$ due to \cref{obs:3}, we get 
    \[
    |\tilde P(r^{-1}w)|=(r^{-1})^{m-1}|\tilde Q(r^{-1}w)|=\frac 1 {r^{m-1}}|\tilde Q(rw)|=\frac 1 {r^{2m-2}} |\tilde P(rw)|>0,
    \]
    where the final inequality follows from the first case of $z=rw$ for $r< 1$. 
\end{proof}

\subsection{Proof of \cref{thm:step2}}
\label{proof:step2}

\begin{proof}
    Let $m_0\geq 100$ and $\beta\geq 4$ for now and their values will be made large enough later. Throughout, we assume that $m\geq m_0$, $\{\theta_k\}_{k=1}^s\subset\T$, and $\tilde\bfU$ are arbitrary such that $\Delta \geq 2\pi \beta/m$ and $\rho:=\|\bfU\bfU^*-\tilde\bfU\tilde\bfU^*\|_2\leq \epsilon$ where $\epsilon$ will be made sufficiently small later. 

        Let $R,\epsilon_1$ be the absolute constants in \cref{thm:step1} and $\alpha,\epsilon_2$ be the absolute constants in \cref{lem:coarseroots}. For $\epsilon\leq \min\{\epsilon_1,\epsilon_2\}$, the Root-MUSIC polynomial $\tilde P$ only has two roots with reflection symmetry in both $D(w_k,R\rho/m)$ and $D(w_k,\alpha/m)$. 
    
        Make $m_0$ even larger if necessary and let $\tau_1,\alpha_1,\epsilon_3$ be the quantities in \cref{obs:5}, which are absolute constants since they only depend on the fixed absolute constants $\alpha, R$. We use the shorthand notation $W_k:=W(w_k,1-\tau_1/m, \alpha_1/m)$. Making $\epsilon$ even smaller if necessary so that $\epsilon\leq \epsilon_3$, by the lemma, 
        $$
        D(w_k,R\rho/m)
        \subset D(w_k,R\epsilon/m)
        \subset W_k
        \subset D(w_k,\alpha/m).
        $$
        Hence, $\tilde P$ has exactly two roots with reflection symmetry in the wedge $W_k$ for all $k=1,\dots,s$. 
    
        For the same $m_0$, set $\alpha_2:=\alpha_1$, make $\beta$ even larger if necessary, and let $\epsilon_4,\tau_2$ be quantities in \cref{lem:far2}. Note $\epsilon_4,\tau_2$ are absolute constants since they only depend on the fixed absolute constants $m_0$ and $\alpha$. By making $\epsilon$ smaller if necessary so that $\epsilon\leq \epsilon_4$, then the lemma ensures that $|\tilde P(z)| >0$ for all $z\in \C$ such that $1-\tau_2/m \leq |z|\leq (1-\tau_2/m)^{-1}$ and $\min_k d(\arg(z),\theta_k)\geq \alpha_2/m$.  
            
    We are ready to complete the proof. Set $\tau := \min\{\tau_1,\tau_2\}$, which is an absolute constant. Next, make $\epsilon$ even smaller if necessary depending on the absolute constant $\tau$ such that 
    $$
    D(w_k,R\rho/m)
    \subset D(w_k,R\epsilon/m)
    \subset \big\{z\in \C \colon 1-\tau_2/m \leq |z|\leq (1-\tau_2/m)^{-1}\big\}. 
    $$
    This completes the proof. 
\end{proof}

\subsection{Proof of \cref{lem:stochasticWedin}}
\label{proof:stochasticWedin}

\begin{proof}
    We first transform the model into the setting considered in \cite[Theorem 6]{dirksen2025subspace}, which applies to the right singular space and assumes the noise has entries with unit variance. Since $\tilde \bfY=\bfPhi \bfA+\bfN$, this is equivalent to 
    $$
    \frac 1 {\sigma} \tilde \bfY^*  = \frac 1 {\sigma} \bfA^* \bfPhi^* + \frac 1 {\sigma} \bfN^* \in \C^{n\times m}.
    $$
    Also note that the leading $s$ right singular vectors of $\tilde\bfY^*$ is precisely $\tilde \bfU$ defined by \eqref{eq:defUtilde2}. Additionally, the random vector $\sigma^{-1} \bfeta_\ell$ has entries which are of unit variance and is $K$-subgaussian. In the noiseless case, the leading $s$ right singular vectors is $\bfU$. Next, we recall that $\sigma_s(\bfPhi)\gtrsim \sqrt m$ when $m\Delta \geq 8\pi$ from \cite{aubel2019vandermonde}. Then
    $$
    \sigma_s^2\left( \bfA^* \bfPhi^* \right) 
    \geq \sigma_s^2(\bfPhi) \sigma_s^2\left( \bfA \right) 
    = n \sigma_s^2(\bfPhi) \lambda_s(\bfSigma_{\bfa})
    \geq C m n \lambda_s(\bfSigma_{\bfa}).
    $$
    
    We are in position to apply \cite[Theorem 6]{dirksen2025subspace}, which tells us there are constants $c_K,\alpha_K$ that only depend on $K$ such that for all $\alpha \geq \alpha_K$, with probability at least $1-e^{-c_K \alpha m}$,
    \begin{align*}
    \left\|\bfU\bfU^*- \tilde \bfU \tilde \bfU^* \right\|_2^2
    &\lesssim \frac{ \alpha m (\sigma_s^2\left( \sigma^{-1} \bfA^* \bfPhi^* \right) + n)} {\sigma_s^4\left( \sigma^{-1} \bfA^* \bfPhi^* \right)}
    = \frac{\alpha m \sigma^2}{\sigma_s^2\left( \bfA^* \bfPhi^* \right)} + \frac{\alpha m n \sigma^4}{\sigma_s^4\left( \bfA^* \bfPhi^* \right)}.
    \end{align*}
    From here, we use the previous inequalities and the assumption $m\lambda_s(\bfSigma_{\bfa})\geq \sigma^2$ to obtain
    \begin{align*}
    \left\|\bfU\bfU^*- \tilde \bfU \tilde \bfU^*\right\|_2^2
    &\lesssim \frac {\alpha\sigma^2 } {n\lambda_s(\bfSigma_{\bfa})} + \frac {\alpha\sigma^4} {m n \lambda_s^2(\bfSigma_{\bfa})}
    \lesssim \frac {\alpha\sigma^2 } {n\lambda_s(\bfSigma_{\bfa})}.
    \end{align*}
    This completes the proof. 
\end{proof}

\subsection{Proofs of observations}
\label{proof:observations}

\begin{proof}[Proof of \cref{obs:3}]
	By \cref{obs:1}, the only zeros of $Q$ on $\partial \D$ are $\{e^{i\theta_k}\}_{k=1}^s$. Since $Q$ is a nonnegative trigonometric polynomial on the unit circle of degree at most $m-1$, by the Fej\'er-Riesz theorem, there is an algebraic polynomial $B$ of degree at most $m-1$ such that factorization \eqref{eq:FR} holds. While $B$ is not unique, we can choose $B$ with a particular property. Note that if $w$ is a root of $B$, then $1/\overline w$ is a zero of $z\mapsto \overline{B(1/\overline z)}$. Hence, we can select $B$ so that its zeros all lie in $\overline \D$, while all zeros of $\overline{B(1/\overline z)}$ lie in the complement of $\D$. 
\end{proof}

\begin{proof}[Proof of \cref{obs:4}]
	Since $\tilde Q$ is nonnegative on the unit circle, by the Reisz-F\'ejer theorem, there is an algebraic polynomial $\tilde B$ of degree at most $m-1$ such that 
	\begin{equation}
		\label{eq:FR2}
		\tilde Q(z)=\tilde B(z)\overline{\tilde B(1/\overline z)}. 
	\end{equation}
	Again, we select $\tilde B$ so that its roots are in $\overline \D$. 
\end{proof}

\begin{proof}[Proof of \cref{obs:P}]
    Recall the second formula for $\tilde Q$ in equation \eqref{eq:defQtilde1}. This yields 
    $$
    \tilde P(z)
    = z^{m-1} \tilde Q(z)
    = z^{m-1} - z^{m-1} \bfphi(1/\overline z)^* \tilde \bfU \tilde \bfU^* \bfphi(z).
    $$
    This formula is still awkward to work with because $1/\overline z$ is not complex differentiable on $\C \setminus \{0\}$, so one cannot directly compute $\tilde P'(z)$ through the chain rule. Recalling definition \eqref{eq:defpsi}, we see that
    $$
    z^{m-1} \overline{\bfphi(1/\overline z)}=z^{m-1}\bfphi(1/z)=\bfpsi(z) \forallspace z\in \C\setminus \{0\}.
    $$
    Combining the two previous formulas shows that formula \eqref{eq:Ptilde1} holds for all $z\in \C\setminus \{0\}$ (but since $P$ is a polynomial, hence continuous at $z=0$, this formula must also hold for $z=0$).
    
    When $\tilde P$ is written in form \eqref{eq:Ptilde1}, it is immediately clear that $\tilde P$ is complex differentiable on $\C$ due to $\bfpsi$ and $\bfphi$ being complex differentiable everywhere. We immediately get formula \eqref{eq:Ptilde2} from chain and product rules. 
\end{proof}

\section{Conclusion}
\label{sec:conclusion}

Perhaps the most important contribution of this paper is a rigorous proof that under natural conditions, the root selection process does not select a perturbed extraneous root, an issue that has not been previously addressed. 

As for the correct roots, this paper provides non-asymptotic estimates through a general mechanism of frequency error $\lesssim$ subspace error$/m$. Additionally, they are minimax optimal for the single snapshot setting, for deterministic and stochastic noise. For the multi-snapshot setting, there is room for improvement as described in \cref{rem:techcomparison}. 

It is possible to obtain explicit values of the constants referenced in the theorems by going through each proof in this paper (and the constants in \cref{lem:far1} can be derived by following the argument in \cite{fannjiang2025optimality}). We made the conscious decision to not keep track of the constants, as that would be tangential to the main purpose of this paper. Additionally, optimization of the constants would require a different and more complicated set of arguments (which are not necessarily insightful) beyond the contents of this paper.

This paper assumes a minimum separation of the frequency set. The performance of Root-MUSIC in the case where frequencies are closer than $\pi/m$, referred to as the super-resolution regime, is unknown. The technical contents of the proofs would fail in this regime, and thus cannot be applied directly to super-resolution. In general, a detailed performance analysis of any type of MUSIC algorithm in the super-resolution regime is an open problem.

\section*{Acknowledgments} 
HH is partially supported by NSF-DMS Award \#2309602. WL is partially supported by NSF-DMS Award \#2309602 and a Cycle 56 PSC-CUNY award. We thank anonymous referees for their careful reading of our manuscript and for their thoughtful comments which have improved the presentation.

\section*{Conflict of Interest}

On behalf of all authors, the corresponding author states that there is no conflict of interest.

\section*{Data Availability Statement}

No data was used in this manuscript. 

\bibliographystyle{plain}
\bibliography{rootMUSICbib}

\end{document}